\documentclass[11pt,letterpaper]{article}

\usepackage[margin=1in]{geometry}
\usepackage{microtype}
\usepackage{graphicx}
\usepackage{caption}
\usepackage{subcaption}
\usepackage{booktabs} % for professional tables
\usepackage[normalem]{ulem}

\usepackage{tikz} 
\usepackage{pgfplots}
\pgfplotsset{compat=1.14}
\usepgfplotslibrary{fillbetween}
\usepackage{hyphenat}
\usepackage{subcaption}

\usepackage{graphicx}
\graphicspath{{figures/}}
\def\input@path{{figures/}}

\renewcommand{\P}[1]{{\mathbb{P}}\left(#1\right)}
\usepackage{enumerate}
\usepackage{amssymb}
\usepackage{amsmath}
\usepackage{nicefrac}
\usepackage{graphicx}
\usepackage{tikz,pgfplots}
\usepackage[breaklinks,bookmarks=false]{hyperref}

\usepackage{algorithm}
\usepackage[noend]{algorithmic}

\usepackage{bbm}
\usepackage[amsmath,amsthm,hyperref,thmmarks]{ntheorem}
\usepackage{adjustbox}
\usepackage{multirow}
\usepackage{xspace}

\usepackage{thm-restate}

\usepackage{natbib}

\newtheorem{theorem}{Theorem}[section]
\newtheorem{corollary}[theorem]{Corollary}
\newtheorem{lemma}[theorem]{Lemma}

\newtheorem{definition}[theorem]{Definition}
\qedsymbol{\openbox}

\usepackage[capitalize, nameinlink]{cleveref}
\crefname{theorem}{Theorem}{Theorems}
\Crefname{lemma}{Lemma}{Lemmas}
\Crefname{claim}{Claim}{Claims}
\Crefname{fact}{Fact}{Facts}
\Crefname{observation}{Observation}{Observations}
\Crefname{invariant}{Invariant}{Invariants}

\DeclareMathOperator{\OPT}{OPT}

\newcommand{\eps}{\varepsilon}

\newcommand{\rb}[1]{\left( #1 \right)} % round Brackets
\newcommand{\scalar}[2]{\langle {#1} ,{#2}\rangle}

\newcommand{\bbRp}{{\mathbb{R}_{\ge 0}}}

\newcommand{\BBAlg}{\textsc{Alg}\xspace}

\newcommand{\cM}{\mathcal{M}}
\newcommand{\cS}{\mathcal{S}}
\newcommand{\cW}{\mathcal{W}}

\newcommand{\marginal}[2]{f\rb{#1 \mid #2}}

\DeclareMathOperator*{\argmin}{argmin}
\DeclareMathOperator*{\argmax}{argmax}
\DeclareMathOperator{\polylog}{polylog}

\newcommand{\omngreedy}{\textsc{Omniscent-Greedy}\xspace}
\newcommand{\omnswap}{\textsc{Omniscent-Swapping}\xspace}
\newcommand{\swapping}{\textsc{Swapping}\xspace}

\newcommand{\robswap}{\textsc{Robust-Swapping}\xspace}

\newcommand{\Sopt}{\OPT}

\newcommand{\E}[1]{{\mathbb{E}}\left[#1\right]}

\pgfplotsset{compat=1.16}

\newcommand{\I}{\mathcal{I}}

\newcommand{\cl}{\text{cl}}

\title{Deletion Robust Submodular Maximization over Matroids
}

\author{
	Paul D{\"u}tting\thanks{Google Research. Email:   
\texttt{\{\href{mailto:duetting@google.com}{duetting}, \href{mailto:ashkannorouzi@google.com}{ashkannorouzi},
		\href{mailto:silviol@google.com}{silviol}, \href{mailto:zadim@google.com}{zadim}\}@google.com}}
	\and
	Federico Fusco\thanks{Department of Computer, Control, and Management 
Engineering ``Antonio Ruberti'', Sapienza University of Rome, 
Italy. Email:   
\texttt{\href{mailto:fuscof@diag.uniroma1.it}{fuscof}@diag.uniroma1.it}.
Part of this work was done while Federico was an intern at Google Research, hosted by Paul D{\"u}tting.}
\and
Silvio Lattanzi{$^*$}
\and
Ashkan Norouzi Fard{$^*$}
\and
Morteza Zadimoghaddam{$^*$}}

\date{January 28, 2022}

\begin{document}
\maketitle

\begin{abstract}
Maximizing a monotone submodular function is a fundamental task in machine learning.  In this paper, we study the deletion robust version of the problem under the classic matroids constraint. Here the goal is to extract a small size summary of the dataset that contains a high value independent set even after an adversary deleted some elements. We present constant-factor approximation algorithms, whose space complexity depends on the rank $k$ of the matroid and the number $d$ of deleted elements.  In the centralized setting we present a $(3.582+O(\eps))$-approximation algorithm with summary size $O(k + \frac{d \log k}{\eps^2})$.  In the streaming setting we provide a $(5.582+O(\eps))$-approximation algorithm with summary size and memory $O(k + \frac{d \log k}{\eps^2})$.  We complement our theoretical results with an in-depth experimental analysis showing the effectiveness of our algorithms on real-world datasets.  
\end{abstract}

\section{Introduction}

Submodular maximization is a fundamental problem in machine learning that encompasses a broad range of applications, including active learning~\citep{GolovinK11} sparse reconstruction~\citep{Bach10,DasK11,DasDK12}, video analysis~\citep{ZhengJCP14}, and data summarization~\citep{lin-bilmes-2011-class,BairiIRB15}.

Given a submodular function $f$, a universe of elements $V$, and a family $\mathcal{F} \subseteq 2^V$ of feasible subsets of $V$, the optimization problem consists in finding a set $S \in \mathcal{F}$ that maximizes $f(S)$. A natural choice for $\mathcal{F}$ are capacity constraints (a.k.a.~$k$-uniform matroid constraints) where any subset $S$ of $V$ of size at most $k$ is feasible. Another standard restriction, which generalizes capacity constraints and naturally comes up in a variety of settings, are matroid constraints. As an example where such more general constraints are needed, consider a movie recommendation application, where, given a large corpus of movies from various genres, we want to come up with a set of recommended videos that contains at most one movie from each genre.

Exact submodular maximization is a NP-hard problem, but efficient algorithms exist that obtain small constant-factor approximation guarantees in both  centralized and streaming setting \citep[e.g.,][]{fisher78-II,CalinescuCPV11,ChakrabartiK15}.

In this work we design algorithms for submodular optimization over matroids that are robust to deletions. 
A main motivation for considering deletions are privacy and user preferences. 
For example, users may exert their ``right to be forgotten'' or may update their preferences and thus exclude some of the data points. For instance, in the earlier movie recommendation example, a user may mark some of the recommended videos as ``seen'' or ``inappropriate,'' and we may wish to quickly update the list of recommendations.

\subsection{The Deletion Robust Approach}
Following \citet{MitrovicBNTC17},
we model robustness to deletion as a two phases game against an adversary. 
In the first phase, the algorithm receives a robustness parameter $d$ and chooses a subset $W \subseteq V$ as summary of the whole dataset $V$. Concurrently, an adversary selects a subset $D\subseteq V$ with $|D| \le d$. The adversary may know the algorithm but has no access to its random bits. In the second phase, the adversary reveals $D$ and the algorithm determines a feasible solution from $W \setminus D$. The goal of the algorithm is to be competitive with the optimal solution on $V \setminus D$. Natural performance metrics in this model are the algorithm's approximation guarantee 
and its space complexity as measured by the size of the set $W$. We consider this problem in both the centralized and in the streaming setting. 

We note that, in this model, to obtain a constant-factor approximation, the summary size has to be $\Omega(k+d)$, even when $f$ is additive and the constraint is a $k$-uniform matroid. To see this, consider the case where exactly $k+d$ of the elements have unitary weight and the remaining elements have weight zero. The adversary selects $d$ of the valuable elements to be deleted, but the algorithm does not know which. To be protective against any possible choice of the adversary, the best strategy of the algorithm is to choose $W$ uniformly at random from the elements that carry weight. This leads to an expected weight of the surviving elements of $|W|\cdot k/(k+d)$, while the optimum is $k$.

Prior work gave deletion robust algorithms for the special case of $k$-uniform matroids. The state-of-the-art  \citep{KazemiZK18} is a $2+O(\eps)$ approximation with $O(k + \frac{d \log k}{\eps^2})$ space in the centralized setting, while in the streaming setting the same approximation is achievable at the cost of an extra multiplicative factor of $O(\frac{\log k}{\eps})$ in space complexity. 

For general matroids, \citet{MirzasoleimanK017} give a black-box reduction, which together with the (non-robust) streaming algorithm of \citet{Ashkan21} yields a $3.147$-approximation algorithm with space complexity $\tilde O(kd)$.\footnote{Where the $\tilde{O}$ notation hides logarithmic factors.} %$O(kd\log\abs{V}^{O(1)})$. 
The multiplicative factor $d$ is inherent in the construction, and the $k$ is needed in any subroutine, so this approach necessarily yields to a space complexity of $\Omega(dk)$. 
The main open question from their work is to design a
\emph{space-efficient} algorithm for the problem. This question is important in many practical scenarios where datasets are large and space is an important resource.

\subsection{Our Results}

We present the first constant-factor approximation algorithms for deletion robust submodular maximization subject to general matroid constraints with almost optimal space usage, i.e., our algorithms only use $\tilde{O}(k+d)$ space. 

More formally, in the centralized setting we present a $(3.582+O(\eps))$-approximation algorithm with summary size $O(k + \frac{d \log k}{\eps^2})$.  In the streaming setting we provide a $(5.582+O(\eps))$-approximation algorithm with summary size and memory $O(k + \frac{d \log k}{\eps^2})$. The constants in the two cases are $2+\beta$ and $4+\beta$, where $\beta$ is $e/(e-1) \approx 1.582$, i.e., the best-possible approximation guarantee for the standard centralized problem \citep{CalinescuCPV11,Feige98}. 
The ``price of robustness'' is thus just an extra additive $2$ or $4$ depending on the setting. At the same time, up to possibly a logarithmic factor, the memory requirements are tight. Finally, note that the state-of-the-art for (non-robust) streaming submodular maximization with matroid constraint is a $3.147$-approximation \citep{Ashkan21}.

Intuitively, the extra difficulty in obtaining space-efficient robust deletion summaries with general matroid constraints is that the algorithm can only use elements respecting the matroid constraint to replace a good deleted element from a candidate solution. This issue gets amplified when multiple elements need to be replaced. This is in sharp contrast to the specal case of $k$-uniform matroids where all elements can replace any other element.

Both our algorithms start by setting a logarithmic number of value thresholds that span the average contribution of relevant optimum elements, and use these to group together elements with similar marginal value. The candidate solution is constructed using only elements from bundles that are large enough (a factor of $1/\eps$ larger than the number of deletions). Random selection from a large bundle protects/insures the value of the selected solution against adversarial deletions. 

In the centralized algorithm, it is possible to sweep through the thresholds in decreasing order. This monotonic iteration helps us design a charging mechanism for high value optimum elements dismissed due to the matroid constraint. We use the matroid structural properties to find an injective mapping from the optimal elements rejected by the matroid property to the set of selected  elements. Given the monotonic sweeping of thresholds, the marginal value of missed opportunities
cannot dominate the values of added elements. 

In the streaming setting, elements arrive in an arbitrary order in terms of their membership to various bundles. Therefore, we keep adding elements as long as a large enough bundle exists. Addition of elements from lower value bundles might technically prevent us from selecting some high value elements due to the matroid constraint. Thus, when considering a new element, we allow for a swap operation with any of the elements in the solution to maintain feasibility of the matroid constraint. We perform the swap if the marginal value of the new element $e$ is substantially (a constant factor) higher than the marginal value that the element $e'$ we are kicking out of the solution had when it was added to the solution. The constant factor gap between marginal values helps us account for not only $e'$ but also the whole potential chain of elements that $e'$ caused directly or indirectly to be removed in the course of the algorithm.

By testing our algorithms on multiple real-world datasets, we validate that they achieve almost optimal values while storing only a small fraction of the elements. In the settings we tried, they typically attain at least $90-95\%$ of the value output by state-of-the-art algorithms that know the deletions in advance even though we only keep a few percents of the elements. Our algorithms persevere in achieving high value solutions and maintaining a concise memory footprint even in the face of large number of deletions. 
\subsection{Related Work}

% Robust submodular optimization has been studied for more than a decade \citep{krause2008robust,orlin2018robust}. 
% \citet{MirzasoleimanK017} provide streaming algorithms that process insertions and deletions on the fly with $\tilde{O}(kd)$ memory for matroid constraints. 
% For cardinality and knapsack constraints, \citet{MitrovicBNTC17}, \citet{KazemiZK18} and \citet{AvdiukhinMYZ19} focused on the deletion robust setting we study and designed centralized, streaming, and distributed algorithms with $\tilde{O}(k+d)$ memory. 
% We have employed some of their interesting techniques  to design deletion robust algorithms for the matroid setting, e.g., value-based bundling and selection by sampling from large pools of candidates.
% Due to lack of space, we refer to Appendix~\ref{sec:app_related_work} for a thorough literature review.

Robust submodular optimization has been studied for more than a decade; in \citet{krause2008robust}, the authors study robustness from the perspective of multiple agents each with its own submodular valuation function. Their objective is to select a subset that maximizes the minimum among all the agents' valuation functions. In some sense, this maximum minimum objective could be seen as a max-min fair subset selection goal.   Another robustness setting that deals with multiple valuation function is distributionally robust submodular optimization \citet{staib2019distributionally} in which we have access to samples from a distribution of valuations functions. These settings are fundamentally different from the robustness setting we study in our paper. 

\citet{orlin2018robust} looked at robustness of a selected set in the presence of a few deletions. In their model, the algorithm needs to finalize the solution before the adversarial deletions are revealed and the adversary sees the choices of the algorithm. Therefore having at least $k$ deletions reduces the value of any solution to zero. Here, $k$ is the cardinality constraint or the rank of the matroid depending on the setting. This is the most prohibitive deletion robust setting we are aware of in the literature and not surprisingly the positive results of \citet{orlin2018robust} and \citet{bogunovic2017robust} are mostly useful when we are dealing with a few number of deletions. 

%dynamic setting
\citet{MirzasoleimanK017} study submodular maximization, and provide a general framework to empower inserting-only algorithms to process deletions on the fly as well as insertions. 
Their result works on general constraints including matroids. 
As a result, they provide dynamic algorithms that process a stream of deletions and insertions with an extra multiplicative overhead of $d$ (on both the computation time and memory footprint) compared to insertion-only algorithms. Here $d$ is the overall number of deletions over the course of the algorithm. 
They propose the elegant idea of running $d+1$ concurrent streaming submodular maximization algorithms where $d$ is the maximum number of deletions. 
Every element is sent to the first algorithm. If it is not selected, it is sent to the second algorithm. If it is not selected again, it is sent to the third algorithm and so on. With this trick, they maintain the invariant that the solution of one of these algorithms is untouched by the adversarial deletions, and therefore this set of $\tilde{O}(dk)$ elements suffice to achieve robust algorithms for matroid constraints. This approach has the drawback of having per update computation time linearly dependent on $d$ which can be prohibitive for large number of deletions. It also has a total memory of $\tilde{O}(dk)$ which could be suboptimal compared to the lower bound of $\Omega(d+k)$. 
%\fnote{Isn't this the same scenario as in \citep{MirzasoleimanK017}?}
%Dynamic setting is another model that captures deletions of elements. Elements are inserted and deleted in a stream of updates, and after every update, the algorithm needs to provide a solution based on the sketch it maintains. 
For closely related dynamic setting with cardinality constraint, \citet{LattanziMNTZ20} provide a $1/2$ approximation algorithm 
for $k$-uniform matroids with per update computation time of poly logarithmic in the total number of updates (length of the stream of updates). This much faster computation time comes at the cost of a potentially much larger memory footprint (up to the whole ground set).
\citet{monemizadeh2020dynamic} independently designed an algorithm with similar approximation guarantee and an update time quadratic in the cardinality constraint with a smaller dependence on logarithmic terms. 
 
As already mentioned, a simple lower bound of $\Omega(d+k)$ exists on the memory footprint needed by any robust constant-factor approximation algorithm. The gap between this lower bound and the $O(dk)$ memory footprint in \citet{MirzasoleimanK017} motivated the follow up works that focused on designing even more memory and computationally efficient algorithms. 

% two phases deletion robust
\citet{MitrovicBNTC17} and  \citet{KazemiZK18} are the first to study the deletion robust setting we consider in our work. They independently
designed submodular maximization algorithms for $k$-uniform matroids. \citet{MitrovicBNTC17} proposed a streaming algorithm that achieves constant competitive ratio with memory $O((k + d) \polylog(k))$. Their results extends to the case that the adversary is aware of the summary set $W$ before selecting the set $D$ of deleted elements. On the other hand, \citet{KazemiZK18} design centralized and distributed algorithms with constant factor approximation and $O(k + d\log(k))$ memory, as well as a streaming algorithm with $2$ approximation factor and memory footprint similar to their distributed result and an extra $\log(k)$ factor. 
We have borrowed some of their ideas including bundling elements based on their marginal values and ensuring that prior to deletions, elements are added to the solution only if their are selected uniformly at random from a large enough bundle (pool of candidates). 

Subsequently, \citet{AvdiukhinMYZ19} showed how to obtain algorithms for the case of knapsack constraints with similar memory requirements and constant factor approximation. While their approximation guarantee are far from optimal, their notion of robustness is strongest than the one we consider: there the adversary can select the set of deleted elements adaptively with respect to the summary produced by the algorithm.  

We aim to achieve the generality of \citet{MirzasoleimanK017} work by providing streaming algorithms that work for all types of matroid constraints while maintaining almost optimal computation time and space efficiency of \citet{MitrovicBNTC17} and \citet{KazemiZK18}.

% sliding window
Sliding window setting is another well studied robustness model. In this case, the deletions occur as sweep through the stream of elements and in that sense they occur regularly rather than in an adversarial manner. Every element is deleted exactly $W$ steps after its arrival. Thus at every moment, the most recent $W$ elements are present and the objective is to select a subset of these present elements. \citet{epasto2017submodular} design a $1/3$ approximation algorithm for the case of cardinality constraints with a memory independent of the window size $W$. \citet{zhao2019submodular} provide algorithms that extend the sliding window algorithm to settings that elements have non-uniform lifespans and leave after arbitrary times.

%Non-robust literature
Prior to deletion robust models and motivated by large scale applications, \citet{MirzasoleimanKSK13} designed the distributed greedy algorithm and showed that it achieves provable guarantees for cardinality constraint problem under some assumptions. \citet{mirrokni2015randomized} followed their work up by providing core-set frameworks that always achieve constant factor approximation in the distributed setting. \citet{barbosa2016new} showed how to approach the optimal $1-1/e$ approximation guarantee by increasing the round complexity of the distributed algorithm. For the case of matroid constraints, \citet{ene2019submodular} provided distributed algorithms that achieve the $1-1/e$ approximation with poly-logarithmic number of distributed rounds. For the streaming setting, the already mentioned \citet{Ashkan21} provided a $0.31$ competitive ratio algorithm with $\Tilde{O}(k)$ memory with a matroid constraint of rank $k$. 

\section{Preliminaries}
We consider a set function $f: 2^V \to \bbRp$ on a (potentially large) ground set $V$. Given two sets $X, Y \subseteq V$, the \emph{marginal gain} of $X$ with respect to $Y$ quantifies the change in value of adding $X$ to $Y$ and is defined as
\[
	\marginal{X}{Y} = f(X \cup Y) - f(Y) \,.
\]
When $X$ consists of a singleton $x$, we use the shorthand $f(x|Y)$ instead of $f(\{x\}|Y)$. The function $f$ is called \emph{monotone} if $\marginal{e}{X}  \geq 0$ for each set $X \subseteq V$ and element $e \in V$, and \emph{submodular} if for any two sets $X$ and $Y$ such that $X \subseteq Y \subseteq V$ and any element $e \in V \setminus Y$ we have 
\[
\marginal{e}{X} \ge \marginal{e}{Y}.
\]
Throughout the paper, we assume that $f$ is given in terms of a value oracle that computes $f(S)$ for given $S \subseteq V$. We also assume that $f$ is \emph{normalized}, i.e., $f(\emptyset) = 0$. We slightly abuse the notation and for a set $X$ and an element $e$, use $X+e$ to denote $X \cup \{e\}$ and $X - e$ for $X \setminus \{e\}$. 

A non-empty family of sets $\cM \subseteq 2^V$ is called a \emph{matroid} if it satisfies the following properties. \textit{Downward-closedness}: if $A \subseteq B$ and $B \in \cM$, then $A \in \cM$; \textit{augmentation}: if $A, B \in \cM$ with $|A| < |B|$, then there exists $e \in B$ such that $A + e \in \cM$. We call a set $A \subseteq 2^V$ \emph{independent}, if $A \in \cM$, and \emph{dependent} otherwise. An independent set that is maximal with respect to inclusion is called a {\em base}; all the bases of a matroid share the same cardinality $k$, which is referred to as the {\em rank} of the matroid. Dually to basis, for any $A\notin \cM$, we can define a {\em circuit} $C(A)$ as a minimal dependent subset of $A$, i.e. $C(A) \notin \cM$ such that all its proper subsets are independent. 

The deletion robust model consists of two phases. The input of the first phase is the ground set $V$ and a robustness parameter $d$, while the input of the second phase is an adversarial set of $d$ deleted elements $D \subset V$, along with the outputs of the first phase. The goal is to design an algorithm that constructs a small size summary robust to deletions $W \subseteq V$ in the first phase, and a solution $S \subseteq W \setminus D$ that is independent with respect to matroid $\cM$ in the second phase. The difficulty of the problem lies in the fact that the summary $W$ has to be robust against {\em any} possible choice of set $D$ by an adversary oblivious to the randomness of the algorithm. 

For any set of deleted elements $D$, the optimum solution denoted by $\OPT(V \setminus D)$ is defined as
\[
\OPT(V \setminus D) = \argmax_{R \subseteq V \setminus D, R \in \cM} f(R).
\]
We say that a two phase algorithm is an $\alpha$ approximation for this problem if for all $ D \subset V \text{ s.t. } |D| \le d, $ it holds that 
\[
    f(\OPT(V \setminus D)) \le \alpha \cdot \E{f(S_D)}, 
\]
where $S_D$ is the solution produced in Phase II when set $D$ is deleted and the expectation is with respect to the eventual internal randomization of the algorithm. An important feature of a two phase algorithm is its summary size, i.e., the cardinality of the set $W$ returned by the first phase.

In this paper, we also consider the streaming version of the problem where the elements in $V$ are presented in some arbitrary order (Phase I) and at the end of such online phase the algorithm has to output a deletion robust summary $W.$ Finally, the deleted set $D$ is revealed and Phase II on $W \setminus D$ takes place offline. The quality of an algorithm for the streaming problem is not only assessed by its approximation guarantee and summary size, but is also measured in terms of its {\em memory}, i.e., the cardinality of the buffer in the online phase.

\section{Centralized Algorithm} \label{section:cent}
In this section, we present a centralized algorithm for the Deletion Robust Submodular Maximization subject to a matroid constraint. To that end, we start by defining some notations. Let $\tilde e$ be the $(d+1)$-th element with largest value according to $f$ and $\Delta$ its value. Given the precision parameter $\eps$, we define the set of threshold $T$ as follows 
\[
T = \left\{ (1 + \eps)^i\mid \eps \cdot \frac{\Delta}{(1+\eps)k} < (1 + \eps)^i \le \Delta \right\}.
\]
The first phase of our algorithm constructs the summary in iterations using two main sets $A$ and $B$. We go over the thresholds in $T$ in decreasing order and update sets $A, B$, and $V$. Let $\tau$ be the threshold that we are considering, then $B_\tau$ contains any element $e \in V$ such that $f(e | A) > \tau$ and $ A
+ e \in \cM$. These are the high contribution elements that can be added to $A$. As long as the size of $B_\tau> d / \eps$, we choose uniformly at random an element from $B_\tau$, add it to $A$ and recompute $B_\tau$. We observe that  $A$ is robust to deletions, i.e., when $d$ elements are deleted, the probability of one specific element in $A$ being deleted is intuitively at most $\eps$. Moreover, since each element added to $A$ is drawn from a pool of elements with {\em similar} marginals, the value of this set after the deletions decreases at most by a factor $(1-\eps)$ in expectation. As soon as the cardinality of $B_{\tau}$ drops below $d/\eps$, we can no more add elements from it directly to $A$ while keeping $A$ robust and feasible. Therefore, we remove these elements from $V$ and save them for Phase II so that they can be used if they are not deleted. During the execution of the algorithm we need to take special care of the top $d$ element with highest $f$ values. To avoid complications, we remove them from the instance before starting the procedure and add them to set $B$ at the end. This does not affect the general logic and only simplifies the presentation and the proofs.

\begin{algorithm}[t]
\caption{Centralized Algorithm Phase I} \label{alg:centralized-phase-I}
\begin{algorithmic}[1]
\STATE \textbf{Input:} Precision $\varepsilon$ and deletion parameter $d$
\STATE $\Delta \gets $ $(d+1)^{th}$ largest value in $\{f(e) \mid e \in V\}$
\STATE $V_d \gets $ elements with the $d$ largest values
\STATE $V \gets V\setminus V_d$, $A \gets \emptyset$
\STATE $T = \left\{ (1 + \eps)^i\mid \eps \cdot \frac{\Delta}{(1+\eps)k} < (1 + \eps)^i \le \Delta \right\}$
\FOR{$\tau \in T$ in decreasing order}
    \STATE $B_{\tau} \gets \{e \in V  \mid A + e \in \cM, f(e \mid A) \geq \tau \}$
    \WHILE{$|B_\tau| \geq \frac{d}{\eps}$}
    \STATE $e \gets$ a random element from $B_{\tau}$
    \STATE $V \gets V - e$, $A \gets A + e$
    \STATE $B_{\tau} \gets \{e \in V  \mid A + e \in \cM, f(e \mid A) \geq \tau \}$
    \ENDWHILE
    \STATE Remove $B_{\tau}$ from $V$
\ENDFOR
\STATE $B \gets V_d \cup \bigcup_{\tau \in T} B_{\tau}$
\RETURN $A, B$
\end{algorithmic}
\end{algorithm}

The summary $W$ computed at the end of Phase I is composed by the union of $A$ and $B$. The second phase of our algorithm uses as routine an arbitrary algorithm \BBAlg for monotone submodular maximization subject to matroid constraint. \BBAlg takes as input a set of elements, function $f$ and matroid $\cM$ and returns a $\beta$-approximate solution. In this phase we simply use \BBAlg to compute a solution among all the elements in $A$ and $B$ that survived the deletion and return the best among the computed solution and the value of the surviving elements in $A$. We state and prove now our main result in the centralized setting; we defer the proofs of some intermediate steps to \Cref{app:centralized}.

\begin{algorithm}[t]
\caption{Algorithm Phase II} \label{alg:phase-II}
\begin{algorithmic}[1]
\STATE \textbf{Input:} $A$ and $B$ outputs of phase I, set $D$ of deleted elements and optimization routine $\BBAlg$
\STATE $A' \gets A \setminus D$, $B' \gets B \setminus D$
\STATE $\tilde S \gets \BBAlg(A' \cup B')$
\RETURN $S \gets \argmax \{f(A'), f(\tilde S)\}$
\end{algorithmic}
\end{algorithm}

\begin{theorem}
\label{thm:robust-centralized}
    For $\eps \in (0, 1/3)$, Centralized Algorithm (\cref{alg:centralized-phase-I} and \cref{alg:phase-II}) is in expectation a $(2+ \beta + O(\eps))$-approximation algorithm with summary size $O(k + \frac{d \log k}{\eps^2})$, where $\beta$ is the approximation ratio of the auxiliary algorithm \BBAlg.
\end{theorem}
\begin{proof}
We start by analyzing the size of the summary that our algorithm returns: the two sets $A$ and $B$. Set $A$ is independent set in matroid $\cM$ so its size is no more than the rank of $\cM$, therefore, $|A| \leq k$. Set $B$ is the union of $V_d$ and $\bigcup_{\tau \in T} B_{\tau}$. Each set $B_{\tau}$ has at most $\frac{d}{\eps}$ element and there are at most $\frac{\log k}{\eps}$ such sets. Set $V_d$ contains $d$ elements. Therefore
\[
|A| + |B| \leq d + k + \frac{d \log k}{\eps^2}\,.
\]
We focus now on bounding the approximation guarantee of our approach. To that end, we fix any set $D$ with $|D| \le d$ and bound the ratio between the expected value of $f(S)$ and $f(\OPT)$ (we omit the dependence on $D$ since it is clear from the context). Let $A' = A \setminus D$ and $B' = B \setminus D$. We also define
\[
C = \left\{e \in V \setminus D \mid  f(e\mid A) \geq \eps \cdot \frac{ f(\OPT)}{k}\right\}.
\]
Intuitively, $C$ contains all the {\emph{important}} elements: the remaining elements do not increase the value of the submodular function considerably if added to $A$. Therefore, we do not lose much by ignoring them. By submodularity and monotonicity, we have that 
\begin{align}
    \nonumber
    f(\OPT) \le& \E{f(\OPT \cup A)} \\
    \nonumber
    \le & \E{f(A)} + \E{f(\OPT \cap C \mid A)}
    %\nonumber
     + \E{f(\OPT \setminus C \mid A)} \\
    \nonumber
    \le & \E{f(A)} + \E{f(\OPT \cap C \cap B' \mid A)} \\
    \label{eq:partition_OPT}
    &+  \E{f((\OPT \cap C) \setminus B' \mid A)}  
    + \E{f(\OPT \setminus C \mid A)}
\end{align}

Recall that $S$ is the solution returned by \cref{alg:phase-II}, we now use its expected value to bound four terms in the last inequality, starting from $\E{f(A)}$. Intuitively, for any element that is part of $A$ the probability of it being deleted is $O(\eps)$, since it is sampled from $d/\eps$ {\em similar} elements uniformly at random and only $d$ elements are deleted. We formalize this idea in \cref{lem:robustness} and show that 
\begin{align} \label{eq:main-offline-1}
    \E{f(A)} \le (1+3\eps )\E{f(A')} \le (1+3\eps ) \E{f(S)}.
\end{align}

For the second term, observe that $\OPT \cap C \cap B' \in \cM$ and is contained in the set of elements passed to \BBAlg, so its value is dominated by $\beta$ times the value of the of $\BBAlg(A' \cup B')$, all in all:
\begin{align} 
    \E{f(\OPT \cap C \cap B' \mid A)} \le
    \beta \cdot \E{f(S)}. \label{eq:main-offline-2}
\end{align}
Bounding the third term is more involved. The goal is to show that  \begin{align} \label{eq:main-offline-3}
\E{f((\OPT \cap C) \setminus B' \mid A)} \leq (1+\eps) \E{f(A)}.
\end{align} 
This statement basically bounds the approximation guarantee of our algorithm in case there are no deletions as well, showing that it is almost a $2$-approximate in that case. The formal argument is provided in \cref{lemma:offline-third-eq} and here we only provide an intuitive proof. 
The elements in $(\OPT \cap C) \setminus B'$ have high contribution with respect to $A$ (definition of set $C$) but are not in $B'$, therefore they were discarded in \Cref{alg:centralized-phase-I} due to the matroid constraint. It is possible to show that in this situation, for each element $e \in (\OPT \cap C) \setminus B'$ there exists a unique element $e' \in A$ such that $f(e | A) \leq (1+\eps)f(e'|A_{e'})$ for a suitable subset $A_{e'}$ of $A$. By a careful telescopic argument, it is possible to finally derive \cref{eq:main-offline-3}. Note that $\E{f(A)}$ has already been bounded in \Cref{eq:main-offline-1}.

The fourth term refers to at most $k$ elements and can be bounded based on the definition of $C$: any element $e$ outside of $C$ is such that $f(e|A) \leq \eps \cdot \frac{f(\OPT)}{k},$ by submodularity we have then that
\begin{align}
\label{eq:main-offline-4}
% \[
    \E{f(\OPT \setminus C \mid A)} \le \sum_{e \in \OPT \setminus C} f(e\mid A) \le \eps \cdot f(\OPT).
% \]
\end{align}
 Plugging \cref{eq:main-offline-1,eq:main-offline-2,eq:main-offline-3,eq:main-offline-4} into \cref{eq:partition_OPT} we obtain
\begin{align*}
    f(\OPT) &\le \E{f(A)} + \E{f(\OPT \cap C \cap B' \mid A)}
+   \E{f((\OPT \cap C) \setminus B' \mid A)}  
+ \E{f(\OPT \setminus C \mid A)} \\
&\le \left(2 + \eps \right)\E{f(A)} + \beta \cdot \E{f(S)} + \eps \cdot f(\OPT) \tag*{(Due to \cref{eq:main-offline-2,eq:main-offline-3,eq:main-offline-4})} \\
&\le \left(2 + \eps \right)(1+3\eps)\E{f(A')} + \beta \cdot \E{f(S)} + \eps \cdot f(\OPT) \tag*{(Due to \cref{eq:main-offline-1})}\\
&\le \left[\left(2 + \eps \right)(1+3\eps) + \beta\right] \cdot \E{f(S)} + \eps \cdot f(\OPT) \tag*{(Definition of $S$)}
\end{align*}
Rearranging terms, we get
\begin{equation*}
    f(\OPT) \le \frac{\left(2 + \eps \right)(1+3\eps) + \beta}{1-\eps} \cdot \E{f(S)} \le \left[ 2 + \beta +  \left(2 \beta + 15 \right) \cdot \eps \right] \cdot \E{f(S)},
\end{equation*}
where in the last inequality we used that $(1-\eps)^{-1} \le (1+2\eps)$ for all $\eps \in (0,0.5)$ and that $\frac{(2 + \eps )(1+3\eps)}{1-\eps} \le 2 + 15 \eps$ for all $\eps \in (0,\frac 13)$. The theorem then follows since as long as $\beta$ is a constant we have $\left(2 \beta + 15 \right) \cdot \eps \in O(\eps)$.
\end{proof}

As consequences of the previous Theorem, we get an approximation factor $4$ if we use as subroutine the greedy algorithm \citep{fisher78-II} or $2 + \tfrac{e}{e-1} \leq 3.58$ if we use continuous greedy as in \citet{CalinescuCPV11}. This last result, together with a better analysis of the dependency on $\eps$, is summarized in the following corollary.

\begin{corollary}
\label{cor:centralized}
    Fix any constant $\delta \in (0,1)$, then there exists a constant $C_{\delta}$ such that for any $\eps \in (0,\delta)$, in expectation a ${(3.582+\eps \cdot C_{\delta})}$-approximation algorithm with summary size $O(k + \frac{d \log k}{\eps^2})$ exists.
\end{corollary}

\begin{proof}
    The first step in the proof of \Cref{thm:robust-centralized} where we use that $\eps$ is bounded below $1/3$ is when we use \Cref{eq:main-offline-1}. \Cref{eq:main-offline-1} is proved in \Cref{lem:robustness}, where a similar result, i.e., \Cref{eq:eps_1}, holding for any $\eps \in (0,1)$ is also stated. 
    Using the latter inequality and following the same steps as in the proof of \Cref{thm:robust-centralized}, we obtain that 
    \[
        f(\OPT) \le \left[\left(2 + \eps \right)\frac{1+\eps}{1-\eps} + \beta\right] \cdot \E{f(S)} + \eps \cdot f(\OPT)       
    \]
    We use as optimization routine \BBAlg continuous greedy, therefore we can plug in $\beta = \frac{e}{e-1}$ and, rearranging the terms we obtain
    \begin{align*}
        f(\Sopt) &\le \left(\frac{e}{(e-1)(1-\eps)} + \frac{(2+\eps)(1+\eps)}{(1-\eps)^2}\right) \E{f(S)}\\
        &\le \Big{[}\frac{e}{e-1} + 2  + \underbrace{\frac{8 e - 7 -\delta(2e-1)}{(e - 1) (\delta - 1)^2}}_{C_{\delta}} \cdot \eps \Big{]}\E{f(S)}.
    \end{align*}
    The last inequality can be numerically verified and holds for any  $\eps \in (0,\delta).$
\end{proof}

\section{Streaming Setting} \label{section:streaming}

\begin{algorithm}[t]
\caption{Streaming Algorithm Phase I} \label{alg:streaming-phase-I}
\begin{algorithmic}[1]
\STATE \textbf{Input:} Precision $\varepsilon$, deletion parameter $d.$
\STATE $T \gets \{(1+\varepsilon)^{i} \mid i \in \mathbb{Z}\}$
\STATE $A \gets \emptyset$, $V_d \gets \emptyset$, $B_{\tau} \gets \emptyset$ for all $\tau \in T$ 
\STATE $V_d \gets $ the first $d$ arriving elements
\STATE $\Delta \gets 0$, $\tau_{min} \gets 0$
\FOR{every arriving element $e'$}
    \IF{$f(e') > \min_{x \in V_d} f(x)$}
        \STATE Add $e'$ to $V_d$, pop element $e$ with smallest value
    \ENDIF
    \STATE \textbf{else} $e \gets e'$
    % \ENDIF
    \STATE $\Delta \gets \max\{f(e), \Delta\}$, $\tau_{min} \gets \frac{\varepsilon}{1+\varepsilon} \cdot \frac{\Delta}{k}$
    \STATE Remove from $T$ all the thresholds smaller than $\tau_{min}$ and delete the relative $B_{\tau}$
    \STATE \textbf{if} $\tau_{min} > f(e\mid A)$ \textbf{then} Discard $e$
    % \IF{$\tau_{min} > f(e\mid A)$}
    %     \STATE 
    % \ENDIF
    \STATE Find the largest threshold $\tau \in T$ s.t. $f(e\mid A)\ge \tau$
    \STATE Add $e$ to $B_{\tau}$
    \WHILE{$\exists \tau \in T$ such that $|B_{\tau}| \geq \frac{d}{\varepsilon}$
    % , in decreasing order of $\tau$
    }
        \STATE\label{line:sample}Remove one element $g$ from $B_\tau$ u.a.r. 
        \STATE $w(g) \gets f(g\mid A)$
        \IF{$A + g \in \I$}
            \STATE $A \gets A + g$
        \ELSE
            \STATE $k_g \gets \argmin\{ w(k) \mid k \in C(A + g)\}$
            \IF{$w(g) > 2 \cdot w(k_g)$}
                \STATE $A \gets A + g - k_g$
                \STATE \textbf{Update} $\{B_{\tau}\}$ according to $A$
            \ENDIF
        \ENDIF
    \ENDWHILE
\ENDFOR
\STATE $B \gets V_d \cup_{\tau \in T} B_{\tau}$
\RETURN $A, B$
\end{algorithmic}
\end{algorithm}

In this section we present our algorithm for Deletion Robust Submodular Maximization in the streaming setting. In this setting, the elements of $V$ in the first phase arrive on a stream and we want to compute the summary $W$ with limited memory. 
Our approach consists in carefully mimicking the swapping algorithm \citep{ChakrabartiK15} in a deletion robust fashion; to that end we maintain an independent candidate solution $A$ and buckets $B_{\tau}$ that contain {\em small} reservoirs of elements from the stream with similar marginal contribution each element with respect to the current solution $A.$ Beyond the $B_{\tau}$, an extra buffer $V_d$ containing the best $d$ elements seen so far is kept. 

Before explaining how these sets are updated in streaming setting, let us elaborate two of the challenges that we face. The first phase of the centralized algorithm, Algorithm~\ref{alg:centralized-phase-I}, recomputes the sets $B_{\tau}$ every time that an element is selected to be added to set $A$. This recomputation is very powerful since it ensures that all elements that can be added to $A$ without violating the matroid constraint and have high marginal gain with respect to $A$ are added to $B$. Therefore, we process them in the order of their contribution. This cannot be achieved in the streaming setting as we cannot keep all the elements and the order that the elements in the stream is not depend on their marginal gain. Moreover the set $A$ is changing overtime and as a results the contribution of the element changes as well. Therefore, keeping set $B$ up to date is challenging in streaming setting. Furthermore, the changes in $A$ can be problematic for the elements in $A$ as well. Consider the case that we add elements to set $A$ for lower value thresholds based on the elements of the stream. Afterwards, there are elements that can be added to higher value threshold. In this scenario, these elements cannot be skipped since they are very valuable and any {\emph{good}} solution needs them. Therefore, based on our previous approach we need to add them to $A$ (even if adding them violates the matroid constraint) and remove some elements to keep set $A$ independent. Notice that these removals change the contribution of the rest of the elements hence some low contribution elements with respect to $A$ can have high contribution after removals. 

We start explaining the algorithm by defining the thresholds $\tau$ that set $B_\tau$ is stored. In the streaming setting, we do not have an a priori estimate of $\OPT$, so that we do not know upfront which are the thresholds corresponding to {\em high quality} elements. This issue is overcome by initially considering {\em all} powers of $(1+\eps)$ and progressively removing the ones too small with respect to $\Delta$, i.e., the $(d+1)$-st largest value seen so far.
Every new element that arrives is first used to update $V_d$:  if its value is smaller then the minimum in $V_d$ then nothing happens, otherwise it is swapped with the smallest element in it. Then, the value of $\Delta$ is updated and all the buckets corresponding to thresholds that are too small are deleted, in order to maintain only a logarithmic number of active buckets. At this point, the new element is put into the correct bucket $B_{\tau}$, if such bucket still exists. 
Now, new elements are drawn uniformly at random from the buckets $B_{\tau}$ as long as no bucket contains more than $d/\eps$ elements. These drawn elements are added to $A$ if and only if it is either feasible to add them directly or they can be {\em swapped} with a {\em less important} element in $A$. To make this notion of importance more precise, each element in the solution is associated with a weight, i.e., its marginal value to $A$ when it was first considered to be added to the solution. An element in $A$ is swapped for a more promising one only if the new one has a weight at least twice as big, while maintaining $A$ independent. 

Every time $A$ changes, the buckets $B_{\tau}$ are completely updated so to maintain the invariant that $B_{\tau}$ contains only elements whose marginal value with respect to the current solution $A$ is within $\tau$ and $(1+\eps) \cdot \tau.$
This property is crucial to ensure the deletion robustness: every bucket contains elements that are similar, i.e., whose marginal density with respect to the current solution is at most a multiplicative $(1+\eps)$ factor away. When the stream terminates, the algorithm returns the candidate solution $A$ and $B$, containing $V_d$ and the surviving buckets. As in the centralized framework, $A$ and $B$ constitute together the deletion robust summary $W$ to be passed to \Cref{alg:phase-II}. The pseudocode of this algorithm is presented in \cref{alg:streaming-phase-I}. 
\begin{theorem}
\label{thm:robust-streaming}
    For $\eps \in (0, 1/3)$, Streaming Algorithm (\cref{alg:streaming-phase-I} and \cref{alg:phase-II}) is in expectation a $(4+ \beta + O(\eps))$-approximation algorithm with summary size and memory $O(k + \frac{d \log k}{\eps^2})$, where $\beta$ is the approximation ratio of the auxiliary algorithm \BBAlg.
\end{theorem}
\begin{proof}
    We start by bounding the memory of the algorithm. Sets $A$ and $V_d$ always contains at most $k$, respectively $d$, elements. Every time a new element of the stream is considered, all the active $B_{\tau}$ contain at most $d/\eps$ elements; furthermore there are always at most $O(\log(k)/\eps)$ of them. This is ensured by the invariant that the active thresholds are smaller than $\Delta$ (by submodularity) and larger than $\tau_{\min}.$ Overall, the memory of the algorithm and the summary size is $O(k + d\log k/\eps^2)$.
    
    As in the analysis of \Cref{thm:robust-centralized}, we now fix any set $D$ and study the relative expected performance of our algorithm. To do so, we need some notation: let $K$ be the set of all elements removed from the solution $A$ at some point; moreover, for any element $g$ added to the solution, let $A_g$ denote the candidate solution when $g$ was added (possibly containing the element $k_g$ that was swapped with $g$). Given $A$ and $K$ we can define the set of important elements:
    \[
        C = \left\{e \in V \setminus D \mid  f(e \mid A \cup K) \geq \eps \cdot \frac{ f(\OPT)}{k}\right\}.
    \]
    Arguing similarly to the centralized case,
    we have that 
    \begin{equation}
    \label{eq:str-small-els}
        f(\OPT \setminus C \mid A \cup K) \le \eps \cdot f(\OPT).
    \end{equation}
    All the elements of the stream that were deleted because their marginal with respect to the current solution was smaller than $\tau_{min}$ are not in $C$ due to submodularity (this is also why we consider the marginal with respect to $A \cup K$ and not simply $A$ as in the centralized case).
    We first observe the following two properties hold (proofs in the appendix):
    \begin{enumerate}
        \item[(i)] $w(K) \le w(A) \le f(A)$ (\Cref{lem:killed_elements})
        \item[(ii)] $f(A \cup K) \le w(A \cup K)$  (\Cref{lem:w(AcupK)_vs_f(AcupK)})
    \end{enumerate}
    Moreover, the weight function can be also used to argue about the robustness of our algorithm. More formally, in \Cref{lem:streaming_robust} we show that \begin{equation}
    \label{eq:str-robustness}
        \E{w(A)} \le (1+3\eps) \E{f(A')}.
    \end{equation}
    We have all the ingredients to decompose $f(\OPT)$ exploiting the definitions of $K$ and $C$ and monotonicity of $f$:
    \begin{align}
    \nonumber
        f(\Sopt) \le& \E{f(\Sopt \cup A \cup K)} \\
    \nonumber
    =& \E{f(A \cup K)} +\E{f(\Sopt \cap C \cap B'\mid A \cup K)} \\
    \label{eq:str-decomposition}
        &+ \E{f((\Sopt \cap C) \setminus B' \mid A \cup K)} + \E{f(\Sopt \setminus C\mid A \cup K)}
    \end{align}
    The last term has already been addressed in \Cref{eq:str-small-els}, so let's focus on the remaining three.
    Start from the first one. Using properties (i) and (ii) we have
    \begin{align}
    \nonumber
        \E{f(A \cup K)} &\le \E{w(A \cup K)}\tag*{Property (ii)}
        \\ 
        &\le \E{w(A) + w(K)}
        \tag*{Definition of $w(\cdot)$}        \\ 
        &\le 2 \cdot \E{w(A)}
        \tag*{Property (i)}
        \\
        \nonumber
        &\le 2(1+3\eps) \E{f(A')} \tag*{\Cref{eq:str-robustness}}
        \\
        \label{eq:str-AcupK}
        &\le 2(1+3\eps) \E{f(S)}
    \end{align}
    The second term can be bounded similarly to \Cref{eq:main-offline-2}, using the assumption on $\BBAlg$, given that $\Sopt \cap C \cap B' \in \cM$ and is contained in $A' \cup B'$, where $\BBAlg$ achieves a $\beta$ approximation:
    \begin{align}
    \label{eq:str-OPTcapB}
        \E{f(\Sopt \cap C \cap B'\mid A \cup K)} \le \beta \cdot \E{f(S)}. 
    \end{align}
    The proof of the last term is more involved and is provided in \Cref{lem:greedy_matroid}.
    \begin{equation}
    \label{eq:str-swapping}
        \E{f((\Sopt \cap C) \setminus B'\mid A \cup K)} \le 2 \cdot \E{f(S)}.    
    \end{equation}
    Plugging \cref{eq:str-small-els,eq:str-AcupK,eq:str-OPTcapB,eq:str-swapping} into \cref{eq:str-decomposition} we obtain the following:
    \begin{align*}
        f(\Sopt) \le& \E{f(A \cup K)}
        +\E{f(\Sopt \cap C \cap B'\mid A \cup K)} \\
        &+ \E{f((\Sopt \cap C) \setminus B' \mid A \cup K)}+ \E{f(\Sopt \setminus C\mid A \cup K)}\\
        \le& \left[2(1+3\eps) + \beta  + 2\right]\cdot\E{f(S)} + \eps \cdot f(\OPT)
    \end{align*}
    Rearranging terms, we get
    \begin{equation*}
    % \label{eq:calculations_streaming}
        f(\OPT) \le \frac{4+6\eps + \beta}{1-\eps} \cdot \E{f(S)} \le \left[ 4 + \beta +  \left(2 \beta + 15 \right) \cdot \eps \right] \cdot \E{f(S)},
    \end{equation*}
    where in the last inequality we used that $(1-\eps)^{-1} \le (1+2\eps)$ for all $\eps \in (0,0.5)$ and that $\frac{4 + 6\eps }{1-\eps} \le 4 + 15 \eps$ for all $\eps \in (0,\frac 13)$. The theorem follows since as long as $\beta$ is a constant it holds that $\left(2 \beta + 15 \right) \cdot \eps \in O(\eps)$.
\end{proof}

As a consequence of the previous Theorem, we get an approximation factor $6$ if we use as subroutine $\BBAlg$ the greedy algorithm \citep{fisher78-II} or $4 + \tfrac{e}{e-1} \leq 5.582$ if we use continuous greedy as in \citet{CalinescuCPV11}. Similar to the previous section, we wrap up the last result and a better analysis on the $\eps$ term in the following corollary.
% \end{corollary}
\begin{corollary}
\label{cor:streaming}
    Fix any constant $\delta \in (0,1)$, then there exists a constant $G_{\delta}$ such that for any $\eps \in (0,\delta)$, in expectation a ${(3.582+\eps \cdot G_{\delta})}$-approximation algorithm with memory and summary size $O(k + \frac{d \log k}{\eps^2})$ exists.
\end{corollary}

\begin{proof}
    The first step in the proof of \Cref{thm:robust-streaming} where we need that $\eps$ is bounded below $1/3$ is in when we use \Cref{eq:str-robustness}, that is formally proven in \Cref{lem:streaming_robust}. In the same Lemma, it is also present a bound, \Cref{eq:eps3}, which holds for any $\eps \in (0,1)$. Using it and following the exact same steps of the proof of \Cref{thm:robust-streaming}, we obtain that 
    \[
        f(\OPT) \le \left[2\frac{1+\eps}{(1-\eps)} + \beta  + 2\right]\cdot\E{f(S)} + \eps \cdot f(\OPT)    
    \]
    We use as optimization routine \BBAlg continuous greedy, therefore we can plug in $\beta = \frac{e}{e-1}$ and, rearranging the terms we obtain
    \begin{align*}
        f(\Sopt) &\le \left(2\frac{1+\eps}{(1-\eps)^2} + \frac{3e  -2}{(1-\eps)(e-1)}\right) \E{f(S)}\\
        &\le \Big{[}\frac{e}{e-1} + 4  + \underbrace{\frac{9 e - 8 -\delta(5e-4)}{(e-1) (1-\delta)^2}}_{G_{\delta}} \cdot \eps \Big{]}\E{f(S)}.
    \end{align*}
    The last inequality can be numerically verified and holds for any  $\eps \in (0,\delta).$
\end{proof}

\section{Experiments}

    In the experiments we evaluate the performance of our deletion robust centralized and streaming algorithms on real world data.
    Prior to our work, the only deletion robust algorithm for submodular maximization subject to general matroid constraints were algorithms obtained through the black-box reduction of \citet{MirzasoleimanK017}, which runs $\Theta(d)$ copies of a (non-robust) streaming algorithm as a subrountine.
    The natural choice for a practical optimization subroutine would be the \swapping algorithm, i.e., Algorithm 1 in \citet{ChakrabartiK15}. %, for which $m = k.$ 
    Instead of implementing this \robswap algorithm in our experiments, we decided to consider an {\em omniscient} version of the \swapping algorithm, i.e., \omnswap, that knows (and avoids) the elements that are going to be deleted by the adversary. The choice of this benchmark has three reasons. First, \robswap is designed for a more stringent notion of deletion robustness, and it would have been unfair to compare it to our two phases algorithm; second, \robswap requires $\Theta(dk)$ memory and summary size, which rapidly approaches $n$ as we consider large $D$ in the experiments. Finally, note that the approximation guarantees for \robswap and \omnswap are identical.
    Similarly, in the centralized setting we consider an omniscient version of the classic lazy greedy algorithm \citep{fisher78-II,Minoux78}, which we refer to \omngreedy, that runs lazy greedy on the surviving elements $V \setminus D$. For further details on the benchmarks, as well as on the datasets, we refer the reader to \Cref{sec:app_exp}.

     To simulate the adversary, we run the lazy greedy algorithm to find and delete a high value independent set. In this way, we make sure that (i) the deleted set has high value, therefore increasing the difficulty of recovering a high value sketch, (ii) the deletions are ``evenly'' spread according to the matroid constraint, e.g., all partitions of a partition matroid are equally interested. Notice that this does not unfairly affect the baselines as they know the elements that are going to be deleted in advance. All the experiments were run on a common computer, and running them on any other device \emph{would not} affect the results in any way. In the following, we explain the experimental setup (matroid constraint, submodular function and datasets used) and present the results.

    \paragraph{Interactive Personalized Movie Recommendation.}
        Movie recommendation systems are one of the common experiments in the context of submodular maximization  \citep[e.g.,][]{MitrovicBNTC17,Norouzi-FardTMZ18,HalabiMNTT20, AmanatidisFLLR20,AmanatidisFLLMR21}. In this experiment, we have a large collection $M$ of movies from various genres $G_1, G_2, \dots, G_k$ and we want to design a deletion robust recommendation system that proposes to users one movie from each genre. A natural motivation for deleted elements is given by movies previously watched by the user, or for which the user has a strong negative opinion.\footnote{It is clear that this class of deletions cannot be captured by \emph{randomized} deletions and a stronger model is needed.} The recommendation system has to quickly update its suggestions as the user deletes the movies. 
        We use the MovieLens 1M database \citep{movielens16}, that contains 1000209 ratings for 3900 movies by 6040 users. Based on the ratings, it is possible to associate to each movie $m$, respectively user $u$, a feature vector $v_m$, respectively  $v_u$.
        More specifically, we complete the users-movies rating matrix and then extract the feature vectors using a singular value decomposition and retaining the first $30$ singular values \citep{TroyanskayaCSBHTBA01}. Following the literature \citep[e.g.,][]{MitrovicBNTC17}, we measure  the quality of a set of movies $S$ with respect to user $u$ (identified by her feature vector $v_u$), using the following monotone submodular objective function:
        \[
            f_u(S) = (1-\alpha) \sum_{s \in S} \scalar{v_u}{v_s}_+ + \alpha \cdot \sum_{m \in M} \max_{s \in S} \scalar{v_m}{v_s},
        \]
        where $\scalar{a}{b}_+$ denotes the positive part of the scalar product.
        The first term is linear and sum the {\em predicted scores} of user $u$ for the movies in $S$, while the second term has a facility-location structure and is a proxy for how well $S$ {\em covers} all the movies. Finally, parameter $\alpha$ balances the trade off between the two terms; in our experiments it is set to $0.95$. Further details are provided in \cref{sec:app_exp}.
        
    \paragraph{Influence Maximization.}
        In this experiment, we study deletion robust influence maximization on a social network graph \citep{Norouzi-FardTMZ18,HalabiMNTT20}. We use the Facebook dataset from \citet{McAuleyL12} that consists of $4039$ nodes $V$ and $88234$ edges $E$. As measure of influence we consider the simple dominating function, that is monotone submodular 
        \[
        f(S) = |\{v \in V: \exists s \in S \text{ and } (s,v) \in E\}|.
        \]
        The vertices of the dataset are already partitioned into $11$ subsets, the so-called circles, and we consider the problem of selecting at most $8$ vertices while choosing at most one vertex from each circle.\footnote{In the original dataset there are elements belonging to more than one circle, for each one of them we pick one of these circles uniformly at random and associate the element only to it.} This way, the constraint can be modeled as the intersection of a partition and a uniform matroid, which is still a matroid.

\begin{figure}[t]
	\captionsetup[subfigure]{aboveskip=0.5pt}
	\centering
	\begin{subfigure}{.33\textwidth}
		\centering
 		\scalebox{0.33}{\includegraphics{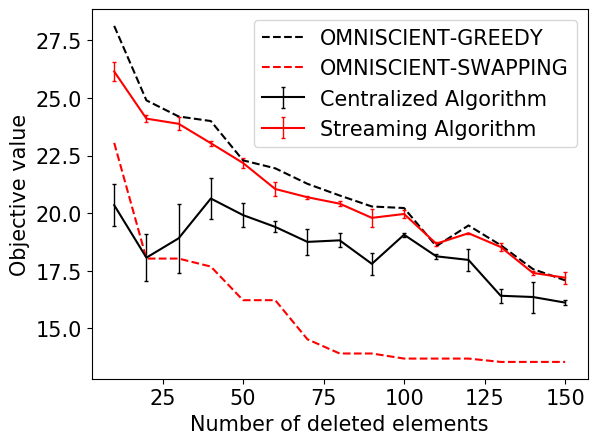}}
		\caption{\footnotesize Recommendation on MovieLens}
		\label{fig:movieLens_obj}
	\end{subfigure}\hspace{0.1pt}%
	\begin{subfigure}{.33\textwidth}
		\centering
		\scalebox{0.33}{\includegraphics{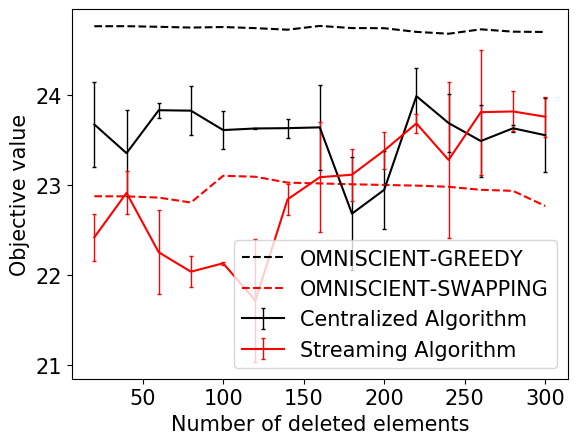}}
		\caption{\footnotesize Kernel log-det on RunInRome}
		\label{fig:run_logdet_obj}
	\end{subfigure}\hspace{0.1pt}%
	\begin{subfigure}{.33\textwidth}
		\centering 		\scalebox{0.33}{\includegraphics{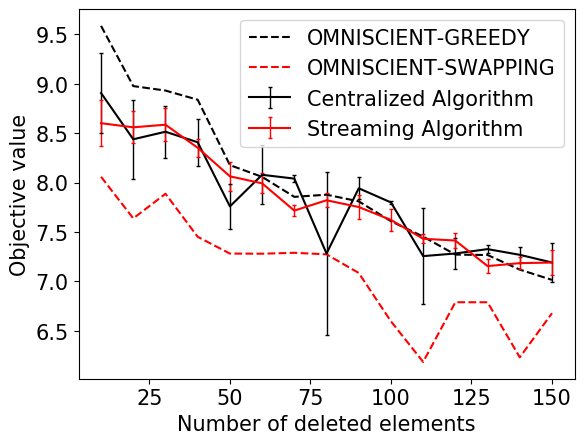}}
		\caption{\footnotesize Kernel log-det on Uber}
		\label{fig:uber_logdet_obj}
	\end{subfigure}
	\vspace{3pt}\\
	\begin{subfigure}{.33\textwidth}
		\centering
 		\scalebox{0.33}{\includegraphics{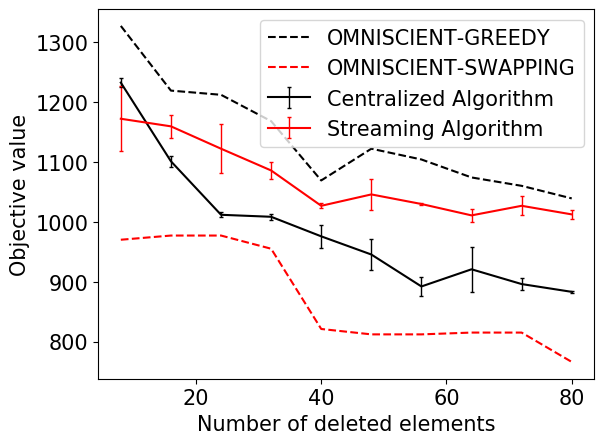}} 
        \caption{\footnotesize Influence Maximization}
		\label{fig:influence_obj}
	\end{subfigure}\hspace{0.1pt}%
	\begin{subfigure}{.33\textwidth}
		\centering
		\scalebox{0.33}{\includegraphics{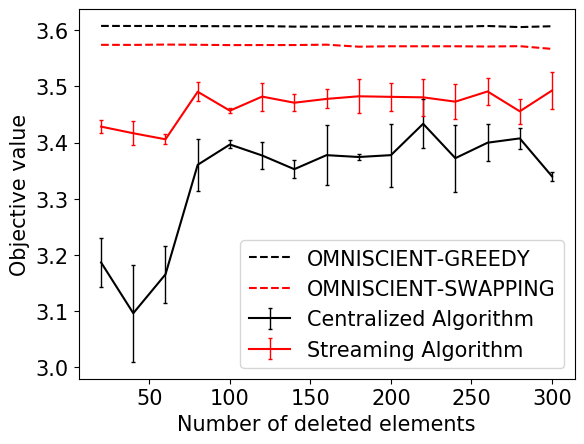}} 
		\caption{\footnotesize K-medoid on RunInRome}
		\label{fig:run_kmedoid_obj}
	\end{subfigure}\hspace{0.1pt}%
	\begin{subfigure}{.33\textwidth}
		\centering
 		\scalebox{0.33}{\includegraphics{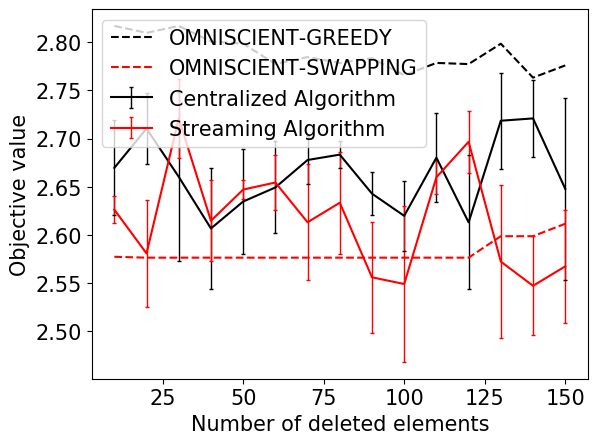}}
		\caption{\footnotesize K-medoid on Uber}
		\label{fig:uber_kmedoid_obj}
	\end{subfigure}
 	\vspace{-5pt}%
	\caption{\small The value of the objective (submodular function $f$) for Centralized and Streaming algorithms compared to the benchmarks \omngreedy and \omnswap with respect to $|D|$. Average and standard deviation over three runs are reported.} 
	\label{fig:objective}
\end{figure}
                
    \paragraph{Summarizing Geolocation Data.}
        The third experimental setting that we consider is the problem of tracking geographical coordinates while respecting privacy concerns. When the sequence of coordinates correspond, for example, to a trajectory, this is inherently a streaming problem, since the coordinates are arriving on a long sequence as the object of interest moves and the goal is to summarizes the path by keeping a short summary.
        To respect the privacy, some of these coordinates cannot be used in the summary and the user can freely choose them. We model these behaviour by deleting these elements from the solution to ensure privacy. In our experiments, we use two datasets. RunInRome \citep{Federico22}, that contains $8425$ positions recorded by running activity in Rome, Italy and a random sample ($10351$ points) from the Uber pickups dataset \citep{Uber14}.
        
        To measure the quality of our algorithms, we consider two different objective functions used in geographical data summarization: the $k$-medoid and the kernel log-det. 
        Consider the $k$-medoid function on the metric set $(V,d)$: 
        \[
            L(S) = \frac{1}{|V|} \sum_{v \in V} \min_{e \in S} d(e,v).
        \]
        By introducing an auxiliary point $e_0 \in V$ we can turn $L$ into a monotone submodular function $f(S) = L(e_0) - L(S + e_0)$ \citep{MirzasoleimanKSK13}. In our experiment we set $e_0$ to be the first point of each dataset.
        
        For the second objective, consider a kernel matrix $K$ that depends on the pair-wise distances of the points, i.e. $ K_{i,j} = \exp\{-\frac{d(i,j)^2}{h^2}\}$ 
        where $d(i,j)$ denotes the distance between the $i^{th}$ and the $j^{th}$ point in the dataset and $h$ is some constant. Following \citet{Krause14,MirzasoleimanK017,KazemiZK18}, a common monotone submodular objective is $f(S) = \log \det (I + \alpha K_{S,S}),$ where $I$ is the $|S|-$dimensional identity matrix, $K_{S,S}$ is the principal sub-matrix corresponding to the entries in $S$, and $\alpha$ is a regularization parameter (that we set to $10$ in the experiments).
        
        The data points of both datasets are partitioned in geographical areas and the goal is to maximize $f$ using at most two data points from each one of them. This can be modeled with a laminar matroid, so our theoretical results apply.

    \paragraph{Experimental Results.}
    In \cref{fig:objective} we present the objective values for our  algorithms with $\varepsilon = 0.99$ (denoted by Centralized Algorithm and Streaming Algorithm) and the benchmarks: \omngreedy and \omnswap. We report the average and standard deviation over three runs.  Choosing $\varepsilon = 0.99$ for our algorithms guarantees that the size of the summary is extremely small, i.e., it is at most $4d$, while its average is around $2d$. Smaller values of $\varepsilon$ result in better solutions but larger summary sizes. The results for a range of $\varepsilon$ are presented in \cref{sec:app_exp}. Across all datasets, our deletion robust algorithm for the centralized setting typically obtains at least 90-95\% of the value of the omniscient benchmark that knows the deletions in advance.
    
    Our deletion robust algorithm for the streaming setting even outperforms its omniscient counterpart in all experiments (up to $20\%$ in some cases), except for \Cref{fig:run_kmedoid_obj}, where it is up to $2\%$ lower. One of the reasons that our streaming algorithm outperforms its benchmark is the extra memory that it is allowed to use. 
    
\section{Conclusion and Future Work}

    We presented the first space-efficient constant-factor approximation algorithms for deletion robust submodular maximization over matroids in both the centralized and the streaming setting. 
    In extensive experiments, we observed that the quality of our algorithms is actually competitive with that of ``all-knowing'' benchmarks that know the deleted elements in advance. A natural direction for future work is to extend these results and ideas to possibly non-monotone objectives, other constraints (e.g., multiple matroids and knapsack), and to consider fully dynamic versions with insertions and deletions.    

    \section*{Acknowledgement}
    Federico Fusco's work was partially supported by the ERC Advanced 
    Grant 788893 AMDROMA ``Algorithmic and Mechanism Design Research in 
    Online Markets'' and the MIUR PRIN project ALGADIMAR ``Algorithms, Games, 
    and Digital Markets.''

    \bibliographystyle{abbrvnat}
    \bibliography{cite}
    \appendix

\section{Proofs Omitted in \cref{section:cent}}
\label{app:centralized}
In this section, we prove the lemmas that are used in \cref{section:cent} and we inherit the notation from that section.
\begin{lemma}
\label{lem:robustness}
    For $\eps \in (0,1)$, we have 
    \begin{equation}
    \label{eq:eps_1}
        \E{f(A)} \le \frac{1+\eps}{1-\eps}\ \E{f(A')}. \tag{\ref{eq:main-offline-1}a}
    \end{equation}
    If $\eps \in (0,\frac 13)$, we the previous inequality implies that
    \begin{equation}
        \E{f(A)} \le (1+3\varepsilon )\E{f(A')}. \tag{\ref{eq:main-offline-1}}
    \end{equation}
\end{lemma}
\begin{proof}
    The proof is similar to that of Lemma 2 in \citet{kazemi17arxiv}. For the sake of the analysis, for each threshold $\tau \in T$, let $A_{\tau}$ denote the set of elements in $A$ that were added during an iteration of the for loop corresponding to threshold $\tau$ in \cref{alg:centralized-phase-I}. Moreover, let $n_{\tau} = |A_{\tau}|$ and order the elements in $A_{\tau} = <x_{1,\tau},x_{2,\tau},\dots, x_{n_{\tau},\tau}>$ according to the order in which they are added to $A_{\tau}.$ Finally, let $I(\ell,\tau)$ be the indicator random variable corresponding to the element $x_{\ell,\tau}$ not being in $D$. Notice that the randomness here is with respect to the random draw of the algorithm in Phase I: $D$ is fixed but unknown to the algorithm. The crucial argument is that $x_{\ell,\tau}$ is drawn uniformly at random from a set of cardinality at least $d/\varepsilon$ where at most $d$ elements lie in $D$, hence 
    \begin{equation}
    \label{eq:robust_proba}
        \P{I(\ell,\tau)} \ge 1- \varepsilon.
    \end{equation}
    As a starting point, we can decompose the value of $A$ as follows using submodularity:
        \begin{align*}
        f(A) &= \sum_{\tau \in T} \sum_{\ell=1}^{n_{\tau}} f(x_{\ell,\tau}|\cup_{t > \tau}A_t \cup \{  x_{1,\tau}, x_{2,\tau}, \dots, x_{\ell-1,\tau}\}).
    \end{align*}
    Recall now that the elements added in a specific iteration of the for loop share the same marginal up to an $(1+\varepsilon)$ factor, i.e.,
    \begin{equation} \label{eq:at_cont}
        \tau \le f(x_{\ell,\tau}|\cup_{t>\tau}A_t \cup \{  x_{1,\tau}, x_{2,\tau}, \dots, x_{\ell-1,\tau}\}) \le \tau \cdot (1+\varepsilon), \quad \forall \tau \in T.        
    \end{equation}

    Summing up those contributions, we have
    \begin{equation}
    \label{eq:robust_ineq}
       \sum_{\tau \in T} |A_{\tau}| \cdot \tau \le f(A) \le (1+\varepsilon) \sum_{\tau \in T} |A_{\tau}| \cdot \tau.  
    \end{equation}
    
    We now decompose in a similar way the value of $A' = A \setminus D$. We let $I(\ell,i) \cdot X$ to be the set $X$ if the indicator variable is $1$ and the empty set otherwise. We also define $A'_{\tau} = A_{\tau} \setminus D = A_{\tau} \cap A'$, we have
    \begin{align*}
        f(A') &= \sum_{\tau \in T} \sum_{\ell=1}^{n_{\tau}} I(\ell,\tau) \cdot f(x_{\ell,\tau}|\cup_{t > \tau}A_t \cup \{I(1,\tau) \cdot x_{1,\tau}, I(2,\tau) \cdot x_{2,\tau}, \dots, I(\ell-1,\tau) \cdot x_{\ell-1,\tau}\})\\
        &\ge \sum_{\tau \in T} \sum_{\ell=1}^{n_\tau} I(\ell,\tau)\cdot f(x_{\ell,\tau}|\cup_{t > \tau}A_t \cup \{x_{1,\tau}, x_{2,\tau}, \dots, x_{\ell-1,\tau}\})\\
        &\ge \sum_{\tau \in T} \sum_{\ell=1}^{n_{\tau}} I(\ell,\tau) \cdot \tau\\
        &=\sum_{\tau \in T} |A'_{\tau}| \cdot \tau,
    \end{align*}
    where the first inequality follows by submodularity and the second one follows from \cref{eq:at_cont}. We apply the expected value to the previous inequality, which results in
    \[
        \E{f(A')} \ge \sum_{\tau \in T} \E{|A'_{\tau}|} \cdot \tau \ge (1-\varepsilon)\sum_{\tau \in T} \E{|A_{\tau}|} \cdot \tau \ge \frac{1-\varepsilon}{1
        + \varepsilon}\E{f(A)},
    \]
    where the second inequality follows by linearity of expectation and \cref{eq:robust_proba}, while the last one from the right hand side of  \cref{eq:robust_ineq}. The Lemma follows observing that $(1+3\varepsilon) \ge \frac{1+\varepsilon}{1
        - \varepsilon}$, for all $\varepsilon \in (0,\frac 13)$.
\end{proof}

\begin{lemma} \label{lemma:offline-third-eq}
    We have,
    \[
        \E{f((\OPT \cap C) \setminus B' \mid A)} \leq (1+\eps) \E{f(A)}.
    \]
\end{lemma}
\begin{proof}
    Let us order the elements in $\Sopt \cap C = \{x_1, x_2, \dots \}$ according to the order in which they were removed from some $B_{\tau}$ because of feasibility constraint (since they are in $C\setminus B$ it must be the case). Notice that once an element fails the feasibility test and is removed from some $B_{\tau}$, then it is never considered again. Furthermore, for each such $x_i$ let $F_i$ be the set $A$ when $x_i$ fails the feasibility test. We know that $F_i$ is independent since $A$ is independent during the execution of \Cref{alg:centralized-phase-I}; moreover, by definition, $F_i + x_i \not \in \cM $. By \cref{lem:hall-lemma} there exists an injective function $h:\Sopt \cap C \to A$ such that $h(x_i) \in F_{i}$ for all $i$. Let $A_i$ denote the set $A$ when the element $h(x_i)$ gets added to it. Notice that $A_i \subseteq F_i$ because $h(x_i)$ has been added to $A$ before $x_i$. We have
    \begin{align*}
        f((\Sopt \cap C) \setminus B'\mid A) &\le \sum_{x \in (\Sopt \cap C) \setminus B'} f(x|A) \\
        &\le \sum_{x_i \in \Sopt \cap C} f(x_i|A_i) \\
        &\le  \sum_{x_i \in \Sopt \cap C} (1+\varepsilon) \cdot f(h(x_i)|A_i) \\
        &\le  (1+\varepsilon) \cdot f(A),
    \end{align*}
    where the first two inequalities follow from submodularity and the fact that $A_i \subseteq A$ for all $i$. The third inequality uses the fact that if $x_i$ was still feasible to add when $h(x_i)$ was added, then their marginals are at most a $(1+\varepsilon)$ factor away. The last inequality follows from a telescopic decomposition of $A$, the monotonicity of $f$ and the fact that $h$ is injective. Applying the expectation to both extremes of the chain of inequalities gives the Lemma.
\end{proof}

\section{Proofs Omitted in \cref{section:streaming}}
\label{app:streaming}
In this section we present the proofs omit in \cref{section:streaming} and we inherit the notations from that section.
\begin{lemma}
\label{lem:killed_elements}
We have 
\begin{itemize}
    \item $w(K) \le w(A) \le f(A)$
    \item $w(A') \le f(A')$.
\end{itemize}
\end{lemma}
\begin{proof}
    The proof of the first inequality is similar to the one of Lemma 9 in \citet{ChakrabartiK15}. Crucially, the weight function $w$ is linear and once an element enters $A$, its weight is fixed forever as the marginal value it contributed when entering $A$. During the run of the algorithm, every time an element $k_g$ is removed from $A$, the weight of $A$ increases by  $w(g) - w(k_g)$ by its replacement with some element $g$. Moreover, $w(k_g) \le w(g) - w(k_g)$ for every element $k_g \in K$ since $2w(k_g) \leq w(g)$. Summing up over all elements in $K$ it holds that 
    \[
          w(K) \le \sum_{k_g \in K}w(k_g) \leq \sum_{k_g \in K} \left[w(g) - w(k_g)\right] \le w(A).
    \]
    
    We show now the second inequality. Let $<a_1, a_2, \dots a_{\ell}>$ be the elements in $A$, sorted by the order in which they were added to $A$. We have that 
    \[
        f(A) = \sum_{i=1}^{\ell} f(a_i|\{a_1,\dots, a_{i-1}\}) \ge \sum_{i=1}^{\ell} f(a_i|A_{a_i}) = \sum_{i=1}^{\ell} w(a_i) = w(A),
    \]
    where $A_{a_i}$ is the solution set right before $a_i$ is added to $A$. The inequality follows from submodularity, since $\{a_1,\dots, a_{i-1}\} \subseteq A_{a_i}$.
    Similarly, let $I(a)$ the indicator random variable corresponding to $a \not \in D$, while $I(a)\cdot a$ is a shorthand for the element $a$ if $I(a)=1$ and the empty set otherwise.
    \begin{align*}
        f(A') &= \sum_{i=1}^{\ell} I(a_i)\cdot f(a_i|\{I(a_1)\cdot a_1,\dots, I(a_{i-1})\cdot a_{i-1}\}) \\
        &\ge \sum_{i=1}^{\ell}  I(a_i)\cdot f(a_i|\{a_1,\dots,a_{i-1}\}) \\
        &\ge \sum_{i=1}^{\ell}  I(a_i)\cdot f(a_i|A_{a_i}) = w(A').
    \end{align*}
\end{proof}

\begin{lemma}
\label{lem:w(AcupK)_vs_f(AcupK)}
    We have
    \[
        f(A \cup K) \le w(A \cup K).
    \]
\end{lemma}

\begin{proof}
    Let $x_1, x_2, \dots x_{t}$ be the elements in $A \cup K$, sorted by the order in which they were added to $A$. We have that 
    \[
        f(A \cup K) = \sum_{i=1}^{t} f(x_i|\{x_1,\dots, x_{i-1}\}) \le \sum_{i=1}^{t} f(x_i|A_{x_i}) = \sum_{i=1}^{t} w(x_i) = w(A \cup K),
    \]
    where $A_{x_i}$ is the solution set right before $x_i$ is added to $A$. The inequality follows from submodularity, since $\{x_1,\dots, x_{i-1}\} \supseteq A_{a_i}$. The reason for the last inclusion is simple: $A_{x_i}$ contains all the elements entered in $A$ before $x_i$ minus those elements that have already been removed from it.
\end{proof}

\begin{lemma}
\label{lem:streaming_robust}
    For $\eps \in (0,1)$, we have 
    \begin{equation}
    \label{eq:eps3} 
        \E{w(A)} \le \frac{1+\eps}{1-\eps}\ \E{f(A')}.\tag{\ref{eq:str-robustness}a}
    \end{equation}
    If $\eps \in (0,\frac 13)$, the previous inequality implies that
    \begin{equation}
    \label{eq:eps4}
     \E{w(A)} \le (1+3\varepsilon )\E{f(A')}.\tag{\ref{eq:str-robustness}}
    \end{equation}
\end{lemma}
\begin{proof}
    Let $A_{\tau}$ be the subset of elements that were added into $A$ coming from $B_{\tau}$ and $A'_{\tau} = A_{\tau} \setminus D$. Moreover, let $a_t^{\tau}$ be the $t^{th}$ element added to $A_{\tau}$ (if any). We have the following:
    \begin{align*}
        \E{|A'_{\tau}|} &= \E{\sum_{t=1}^{|A_{\tau}|} \mathbbm{1}(a_t^{\tau} \not \in D) } = \sum_{t=1}^{+\infty}\E{ \mathbbm{1}(a_t^{\tau} \not \in D) \cdot \mathbbm{1}(|A_{\tau}| \ge t) } \\
        &=\sum_{t=1}^{+\infty} \E{\mathbbm{1}(a_t^{\tau} \not \in D)||A_{\tau}| \ge t} \cdot \P{|A_{\tau}| \ge t} \\
        &\ge (1-\varepsilon)\sum_{t=1}^{+\infty}\P{|A_{\tau}| \ge t}\\
        &=(1-\varepsilon)\E{|A_{\tau}|}.
    \end{align*}
    The crucial observation is now that when the algorithm decides to add an element to $A$ from $B_{\tau}$, then the probability that the element belongs to $D$ is at most $\eps.$
    
    Recall the definition of $A_a$ as the elements in $A$ when $a$ is added, so that $w(a) = f(a|A_a)$. Moreover, $I(a)$ is the indicator variable of the event $a \in A'$ given that $a \in A$. We have
    \[
        w(A) = \sum_{\tau \in T} \sum_{a \in A_{\tau}} w(a)\text{, while } w(A') = \sum_{\tau \in T} \sum_{a \in A_{\tau}} I(a) \cdot w(a) 
    \]
    We know that the weight of an element $a$ coming from $B_t$ is such that $\tau \le w(e) \le (1+\eps)\cdot \tau$. Therefore, we can proceed similarly to what we had in \Cref{lem:robustness}:
    \begin{align*}
        \E{w(A')} &= \E{\sum_{\tau \in T} \sum_{a \in A_{\tau}} I(a) \cdot w(a)} = \sum_{\tau \in T} \E{\sum_{a \in A_{\tau}} I(a) \cdot w(a)} \ge \sum_{\tau \in T} \tau \cdot \E{\sum_{a \in A_{\tau}} I(a)}\\
        &= \sum_{\tau \in T} \tau \cdot \E{|A'_{\tau}|} \ge  (1-\varepsilon) \cdot \sum_{\tau \in T} \tau\cdot \E{|A_{\tau}|} \ge \frac{(1-\varepsilon)}{(1+\varepsilon)} \E{w(A)}
    \end{align*}
    The Lemma follows by \Cref{lem:killed_elements} and that for $\eps \in (0,\frac 13)$ it holds that $(1+3\eps) (1-\eps) \ge (1+\eps)$.
\end{proof}

\begin{lemma}
\label{lem:greedy_matroid}
We have $\E{f((\Sopt \cap C) \setminus B'\mid A \cup K)} \le 2 \cdot \E{f(S)}$.
\end{lemma}
\begin{proof}
    We prove a more general statement. Let $\Gamma$ be the set of all elements that, at some point of the run of the algorithm, were considered in an iteration of line \ref{line:sample}. We have the following: for all $X \subseteq \Gamma,$  $X\in \I$, the following inequality holds true
    \[ 
    w(X) \le 2 \cdot w(A).
    \]
    The Lemma follows from Theorem 1 of \citet{Varadaraja11} (using a single matroid and setting $r = 2$). More in specific, we can imagine to restrict the stream to consider only the elements in $\Gamma$, with the order in which they are considered in line \ref{line:sample}.
\end{proof}

\section{Combinatorial Properties of Matroids}

In this section, we focus on showing combinatorial properties of matroids that are used in our analysis. Similar properties have been used in this line of research. We start by stating the main result of this section. To this end, consider a matroid $\cM = (E,\mathcal{I})$, where $E$ is a finite set (called the \emph{ground set}) and $\mathcal{I}$ is a family of subsets of $E$ (called the \emph{independent sets}). 

\begin{lemma}
\label{lem:hall-lemma}
Consider a matroid $\cM = (E,\mathcal{I})$ and two sets $F \subseteq E$, and $G \in \mathcal{I}$. Suppose that for all $x \in G\setminus F$ there exist a set $F_x \subseteq F, F_x \in \mathcal{I}$ such that $F_x  + x \not\in \mathcal{I}$. Then there exists a mapping $h: G \setminus F \rightarrow F$ such that 
\begin{itemize}
    \item for all $x \in G\setminus F$, $h(x) \in F_x$, and
    \item for all $x,y \in G\setminus F$ with $x \neq y$, $h(x) \neq h(y)$.
\end{itemize}
\end{lemma}

Intuitively, $h$ as a {\emph{semi-matching}} that matches all elements in $G \setminus F$ to an element in $\cup_{x \in G\setminus F} F_x$, while some of the elements in $\cup_{x \in G\setminus F} F_x \subseteq F$ can remain unmatched.

We use Hall's Marriage Theorem to prove the lemma. The combinatorial version of this theorem concerns set systems and the existence of a transversal (a.k.a.~system of distinct representatives). 
Formally, let $\cS$ be a family of finite subsets of a base set $X$ ($\cS$ may contain the same set multiple times). A \emph{transversal} is an injective function $f: \cS \rightarrow X$ such that $f(S) \in S$ for every set $S \in \cS$. In other words, $f$ selects one representative from each set in $S$ in such a way that no two of these representatives are equal.   
\begin{definition}
A family of sets $\cS$ satisfies the \emph{marriage condition} if for each subfamily $\cW \subseteq \cS$,
\[
|\cW| \leq \left|\bigcup_{F \in \cW} F\right| \,.
\]
\end{definition}

\begin{theorem}[Hall (1935)] \label{hall-thm}
A family of sets $\cS$ has a transversal if and only if $\cS$ satisfies the marriage condition.
\end{theorem}

Let us recall two well-known properties of matriods.
\begin{definition}[restriction]
Given a matroid $\cM = (E,\mathcal{I})$ and a set $S \subseteq E$ the \emph{restriction} of $\cM$ to $S$, denoted by $\cM \mid S$, is matroid $\cM' = (S,\mathcal{I'})$ where $\mathcal{I}' = \{T \subseteq S \mid T \in \mathcal{I}\}$.
\end{definition}

\begin{definition}[contraction]
Given a matroid $\cM = (E,\mathcal{I})$ and a set $S \subseteq E$ the \emph{contraction} of $\cM$ by $S$, written $\cM / S$, is the matroid $M'$ on $E\setminus S$ with rank function $r_{\cM'}(T) = r_\cM(T \cup S) - r_\cM(S)$.

Specifically, for $S \in \mathcal{I}$, the restriction of $\cM$ by $S$ is the matroid $\cM' = (E \setminus S, \mathcal{I}')$ where $\mathcal{I}' = \{T \subseteq E \setminus S \mid T \cup S \in \mathcal{I}\}$.
\end{definition}
We are ready to present the proof of the main lemma of this section.
\begin{proof}[Proof of \cref{lem:hall-lemma}]
    Let $V$ be any subset of $G\setminus F$. Define $F_V = \bigcup_{x \in V} F_x$. We want to show that 
    \[
    |F_V| = |\bigcup_{x \in V} F_x| \geq |V|. 
    \]

    To do that, we first show that $\cl(V \cup F_V) = \cl(F_V)$. Recall that $\cl(.)$ denotes the \emph{closure} (or \emph{span}) of a set.
    
    We have that $V \subseteq \cl(F_V)$, in fact, for each element $x \in V$ there exists a subset $F_x \subseteq F_V$ such that $F_x + x \not\in \I$ and therefore $x \in \cl(F_x)$, which implies, by monotonicity of the closure with respect to the inclusion, that $x \in \cl(F_V)$.
    We also know that $F_V \subseteq \cl(F_V)$, hence $F_V \cup V \subseteq \cl(F_V)$. Now, if we apply the closure to both sets, we get
    \[
        \cl(F_V \cup V) \subseteq \cl(\cl(F_V)) = \cl(F_V) \subseteq \cl(F_V\cup V),
    \]
    where the equation follows by the well-known properties of closure. This shows that $\cl(F_V \cup V) = \cl(F_V)$ as claimed.
    
    Now let us look at the restriction of the matroid $\cM$ to $\cl(F_V \cup V) = \cl(F_V)$. Afterwards, contract this matroid by $(F \cap G) \cap \cl(F_V)$. Call this matroid $\cM'$, and denote its rank by $r'$. We claim that $V$ is independent in this new matroid $\cM'$. This is due to, $V \subseteq G \setminus F$ and $V \cup (F \cap G) \subseteq G$ and $G \in \I$.
    
    We thus have 
    \[
    r_\cM(\cl(F_V)) \geq r_\cM(\cl(F_V) \setminus (F \cap B) ) = r' \geq |V|.
    \]
    
    Finally,
    \[
    |F_V| \geq r_\cM(F_V) = r_\cM(\cl(F_V)).
    \]
    Putting these two chains of inequalities together, we obtain $|F_V| \geq |V|$ as claimed. The proof now follows by applying \cref{hall-thm}.
\end{proof}

\begin{minipage}{0.46\textwidth}
\begin{algorithm}[H]
\caption{Lazy Greedy}
\label{alg:lazy_greedy}
\begin{algorithmic}[1]
    \STATE \textbf{Input:} Precision parameter $\eps_0>0$
    \STATE $A \gets \emptyset$, $\Delta \gets \max_{e \in V} f(e)$, max-iter $\gets \frac{1}{\eps_0}\log{\frac{k}{\eps_0}}$
    \STATE Let $Q$ be a priority queue
    \FOR{$e \in V$}
        \STATE Add $e$ to $Q$ with priority $p(e) = f(e)$
        \STATE cont$(e) \gets 0$
    \ENDFOR
    \WHILE{$Q$ is not empty}
        \STATE Pop element $e$ with largest priority $p(e)$ from $Q$
        \IF{$A + e\notin \cM$ or cont$(e) \ge $ max-iter}
            \STATE Discard $e$
        \ELSIF{$p(e) \le (1+\eps_0) \cdot f(e|A)$}
            \STATE Add $e$ to $A$
        \ELSE
            \STATE Add $e$ to $Q$ with priority $p(e) = f(e|A)$
            \STATE Increase cont$(e)$ by $1$
        \ENDIF
    \ENDWHILE
    \RETURN $A$
\end{algorithmic}
\end{algorithm}
\end{minipage}
\hfill
\begin{minipage}{0.46\textwidth}
\begin{algorithm}[H]
\caption{Swapping algorithm} 
\label{alg:swapping}
\begin{algorithmic}[1]
    \STATE $A \gets \emptyset$
    \FOR{every arriving element $e$}
        \STATE $w(e) \gets f(e|A)$
        \IF{$A + e \in \cM$}
            \STATE $A \gets A + e$
        \ELSE
            \STATE $k \gets \argmin\{ w(k) \mid k \in C(A + e)\}$
            \IF{$2 \cdot w(k) < w(e)$}
                \STATE $A \gets A - k + e$
            \ENDIF
        \ENDIF
    \ENDFOR
    \RETURN $A$
\end{algorithmic}
\end{algorithm}
\end{minipage}

\section{More on the Experimental framework}
\label{sec:app_exp}

    In the experiments, we use as subroutines two famous algorithms: the lazy implementation of the greedy algorithm for monotone submodular maximization subject to a matroid constraint \citep{fisher78-II,Minoux78} and the streaming algorithm from \citet{ChakrabartiK15}. For the sake of completeness the pseudocodes are reported in \Cref{alg:lazy_greedy} and \Cref{alg:swapping}, while their main properties are summarized in \Cref{thm:folklore}. We use the lazy implementation of greedy because it is way quicker than simple greedy and loses only a small additive constant in the approximation. For our experiments we set this precision parameter of lazy greedy $\eps_0$ to $0.0001$. 
    
    \begin{theorem}[Folklore]
    \label{thm:folklore}
        The lazy implementation of the Greedy algorithm (\Cref{alg:lazy_greedy}) is a deterministic $(2+O(\eps_0))$-approximation for submodular maximization subject to a matroid constraint. Moreover it terminates after $O(\frac{n}{\eps_0} \log k)$ calls to the value oracle of the submodular function.
        The swapping algorithm (\Cref{alg:swapping}) is a a deterministic $4$-approximation for streaming submodular maximization subject to a matroid constraint. Its memory is exactly $k$.
    \end{theorem}
    
    \paragraph{Similarities between \robswap and \omnswap.} \robswap and \omnswap are both guaranteed to produce a $4$ approximation to the best independent set in $V \setminus D$. Besides, their output share the same basic structure: after the deletions, \omnswap outputs the instance of \swapping with largest index between those that did not suffered deletions. After all, this is just an instance of \swapping that is guaranteed to not have suffered deletions and that has parsed all elements in $V\setminus D$; this is similar to directly running \omnswap and ignoring the elements in $V\setminus D$, the only difference being the order in which the elements of the stream were considered and the fact that \omnswap might have contained some element (later removed) from $D$.

    \paragraph{Interactive Personalized Movie Recommendation.} We consider only movies with at least one rating. The movies in the dataset may belong to more than one genre. Since considering directly as constraint the resulting rule \textit{At most one movie from each genre} would not be a matroid, we associated to each movie a genre distribution vector, then we clustered the movies using k-means++ on those vectors to recover $10$ clusters/macro-genres. In the streaming set up we consider a random permutation of the movies, fixed across the experiment. The feature vector $v_u$ of the user $u$ to whom the reccomendation is personalized is drawn uniformly at random from $[0,1]^{30}$.

    \paragraph{Kernel Log-determinant.} The kernel matrix is defined as $  K_{i,j} = e^{-\left(\frac{d_{i,j}}{h}\right)^2}$,    where $d_{i,j}$ denotes the distance between the coordinates of the $i^{th}$ and $j^{th}$ locations while $h$ is a normalization parameter. We set $h$ to be the empirical standard deviation of the pairwise distances in the RunInRome dataset, while we set $h^2=5000$ in the UberDataset. We mention that the identity matrix in the objective function is needed to be sure to take the log of the determinant of a positive definite matrix. The parameter $\alpha$ then tunes the importance of this regularizing perturbation. We set it to be $10$ for both the datasets.
    
    \paragraph{RunInRome.} In the streaming experiments with this dataset we kept the original order of the positions. The partition of the data points in geographical areas is achieved by dividing the center of Rome in an equally spaced $5 \times 5$ grid (in terms of latitude and longitude). The non-empty cells of the grid are $10.$ 
    
    \paragraph{Uber Dataset.} The positions considered in the experiments have been obtained by uniform sampling with parameter $1/50$ from the pickups locations of April 2014. In the streaming setting we considered the order of the dataset as sampled. We used the {\em base} feature of the dataset to partition the positions into $5$ subsets.

\paragraph{More experiments.} Further experiments for different values of the parameter $\eps$ are reported in \Cref{fig:movieLens,fig:facebook,fig:runInRome_logdet,fig:runInRome_kmedoid,fig:uber_logdet,fig:uber_kmedoid}. Average and standard deviation over three runs are reported. We observe that, as suggested by the theoretical results, the value of the objective function decreases as we increase the value of $\eps$. For instance, in \cref{fig:facebook} the value of the objective function decreases from $\approx 1300$ to $\approx 1200$ when we increase $\eps$ from $0.3$ to $0.99$. We observe that the performances achieved in the experiments are way better than the theoretical worst-case guarantees. One explanation is that in practice the probability that an element added to $A$ gets deleted is way larger than $1-\eps$: in our analysis we consider the pessimistic case where the adversary manages somehow to {\em always} delete $d$ out of the $d/\epsilon$ elements contained in the bucket we are sampling from. This is however in general quite unlikely: the composition of the buckets may change over time and it is possible there is no large intersection between all the buckets that were considered across all the ``sampling'' steps of the algorithm.

\begin{figure*}[ht!]
	\captionsetup[subfigure]{aboveskip=1pt}
	\centering
	\begin{subfigure}{.33\textwidth}
		\centering
 		\scalebox{0.33}{\includegraphics{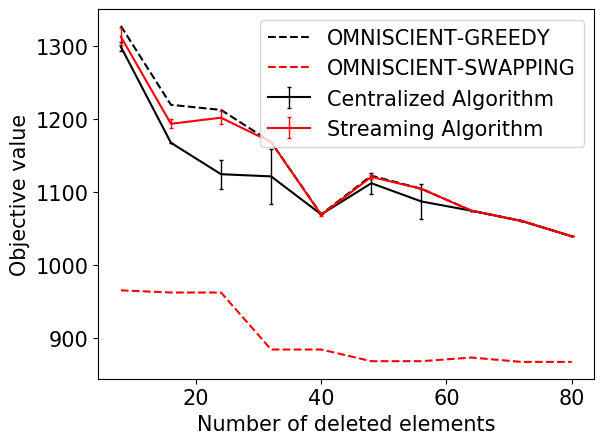}}
		\caption{\footnotesize Results for $\eps = 0.3$}
	\end{subfigure}\hspace{0.1pt}%
	\begin{subfigure}{.33\textwidth}
		\centering
		\scalebox{0.33}{\includegraphics{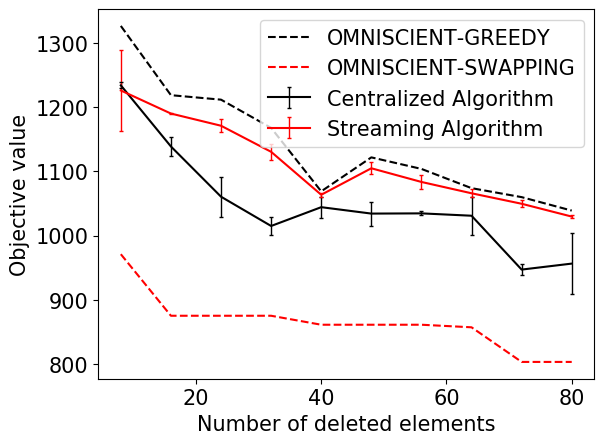}}
		\caption{\footnotesize Results for $\eps = 0.5$}
	\end{subfigure}\hspace{0.1pt}%
	\begin{subfigure}{.33\textwidth}
		\centering 		\scalebox{0.33}{\includegraphics{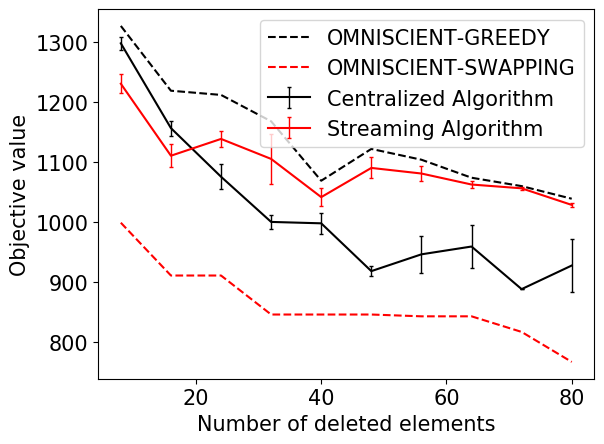}}
		\caption{\footnotesize Results for $\eps = 0.7$}
	\end{subfigure}
	\vspace{5pt}\\
	\begin{subfigure}{.33\textwidth}
		\centering
 		\scalebox{0.33}{\includegraphics{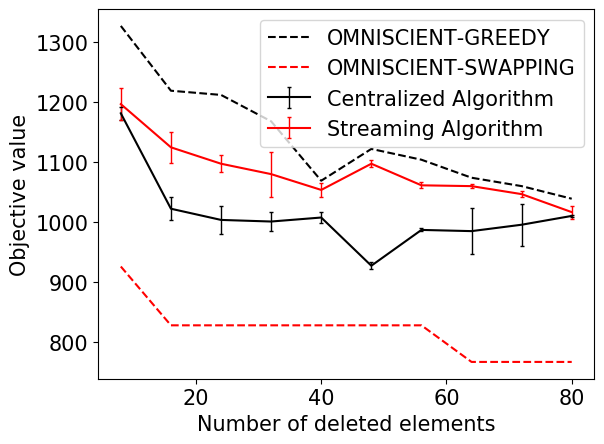}} 
        \caption{\footnotesize Results for $\eps = 0.9$}
	\end{subfigure}\hspace{0.1pt}%
	\begin{subfigure}{.33\textwidth}
		\centering
		\scalebox{0.33}{\includegraphics{facebook/facebook_objective_0.99.png}} 
		\caption{\footnotesize Results for $\eps = 0.99$}
	\end{subfigure}
% 	\vspace{8pt}%
	\caption{\small Results of the Influence Maximization experiments on the Facebook dataset for different values of $\eps$. Note how as $\eps$ decreases the performances of our algorithms improve, while being always comparable with the benchmarks. The lines corresponding to $\omnswap$ changes in different plots since the random permutation considered changes; this however does not change its qualitative performance.} 
	\label{fig:facebook}
\end{figure*}

\begin{figure*}[ht!]
	\captionsetup[subfigure]{aboveskip=1pt}
	\centering
	\begin{subfigure}{.33\textwidth}
		\centering
 		\scalebox{0.33}{\includegraphics{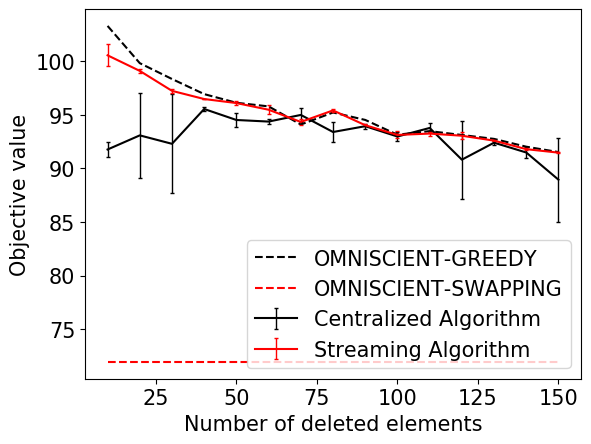}}
		\caption{\footnotesize Results for $\eps = 0.3$}
	\end{subfigure}\hspace{0.1pt}%
	\begin{subfigure}{.33\textwidth}
		\centering
		\scalebox{0.33}{\includegraphics{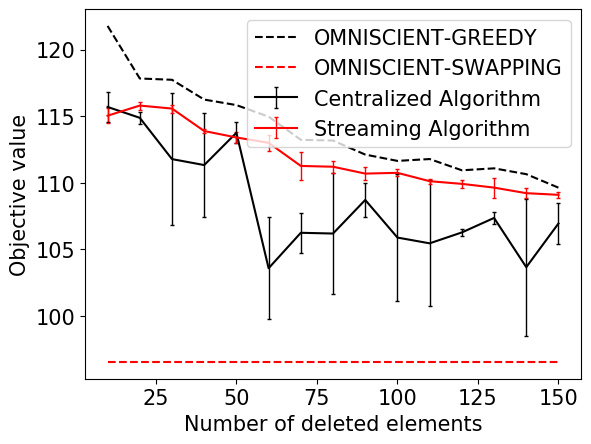}}
		\caption{\footnotesize Results for $\eps = 0.5$}
	\end{subfigure}\hspace{0.1pt}%
	\begin{subfigure}{.33\textwidth}
		\centering 		\scalebox{0.33}{\includegraphics{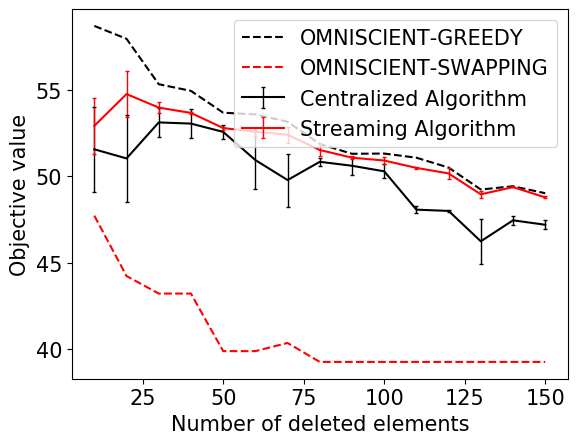}}
		\caption{\footnotesize Results for $\eps = 0.7$}
	\end{subfigure}
	\vspace{5pt}\\
	\begin{subfigure}{.33\textwidth}
		\centering
 		\scalebox{0.33}{\includegraphics{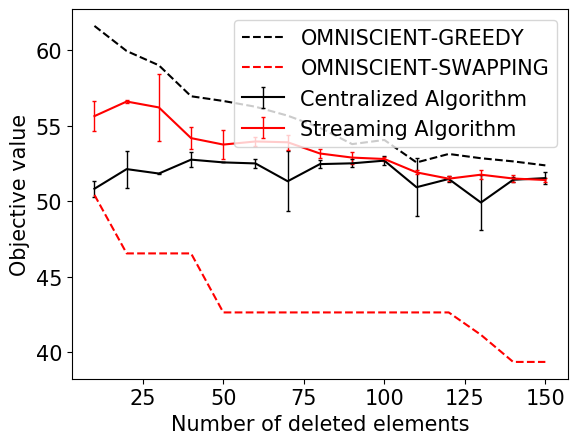}} 
        \caption{\footnotesize Results for $\eps = 0.9$}
	\end{subfigure}\hspace{0.1pt}%
	\begin{subfigure}{.33\textwidth}
		\centering
		\scalebox{0.33}{\includegraphics{movieLens/movieLens_objective_0.99.png}} 
		\caption{\footnotesize Results for $\eps = 0.99$}
	\end{subfigure}
% 	\vspace{8pt}%
	\caption{\small Results of the Interactive Personalized Movie Recommendation on the MovieLens dataset for different values of $\eps$. Note how as $\eps$ decreases the performances of our algorithms improve, while being always comparable with the benchmarks. The performances of the benchmarks change for different values of $\eps$ because in each experiment a new user feature vector is drawn uniformly at random as well as a new permutation is considered in the streaming setting. The results are however qualitatively comparable.} 
	\label{fig:movieLens}
\end{figure*}

\begin{figure*}[ht!]
	\captionsetup[subfigure]{aboveskip=1pt}
	\centering
	\begin{subfigure}{.33\textwidth}
		\centering
 		\scalebox{0.33}{\includegraphics{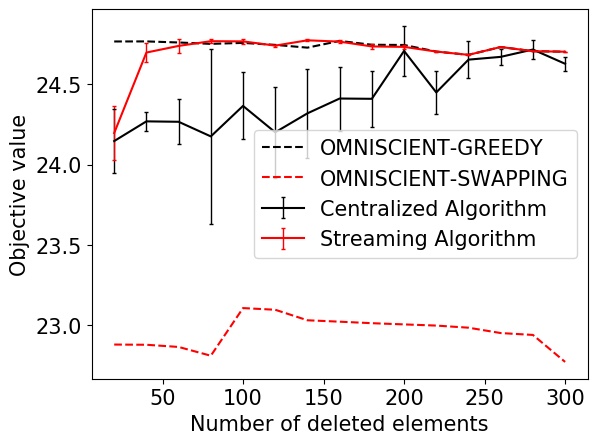}}
		\caption{\footnotesize Results for $\eps = 0.3$}
	\end{subfigure}\hspace{0.1pt}%
	\begin{subfigure}{.33\textwidth}
		\centering
		\scalebox{0.33}{\includegraphics{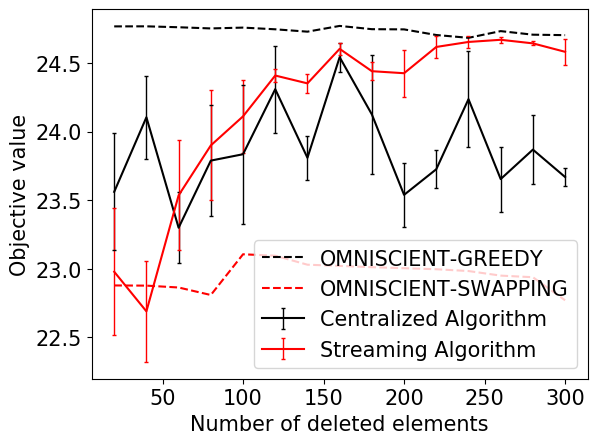}}
		\caption{\footnotesize Results for $\eps = 0.5$}
	\end{subfigure}\hspace{0.1pt}%
	\begin{subfigure}{.33\textwidth}
		\centering 		\scalebox{0.33}{\includegraphics{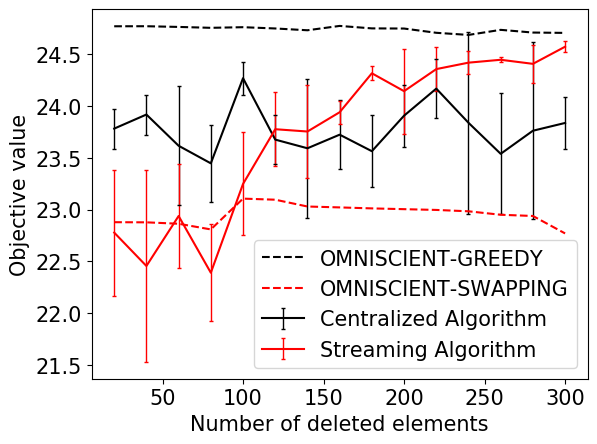}}
		\caption{\footnotesize Results for $\eps = 0.7$}
	\end{subfigure}
	\vspace{5pt}\\
	\begin{subfigure}{.33\textwidth}
		\centering
 		\scalebox{0.33}{\includegraphics{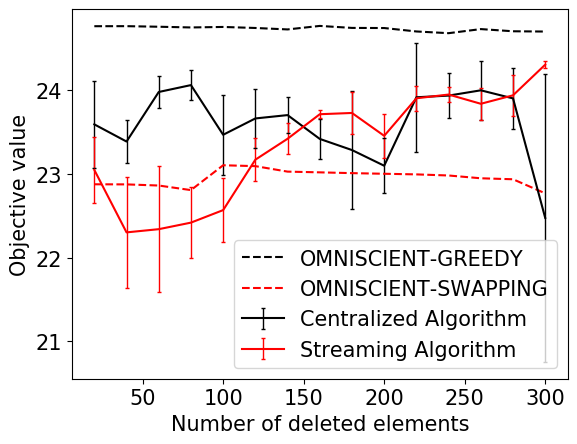}} 
        \caption{\footnotesize Results for $\eps = 0.9$}
	\end{subfigure}\hspace{0.1pt}%
	\begin{subfigure}{.33\textwidth}
		\centering
		\scalebox{0.33}{\includegraphics{runInRome logdet/run_logdet_objective_0.99.png}} 
		\caption{\footnotesize Results for $\eps = 0.99$}
	\end{subfigure}
% 	\vspace{8pt}%
	\caption{\small Results of the kernel logdet experiment on the RunInRome dataset for different values of $\eps$. Note how as $\eps$ decreases the performances of our algorithms improve, while being always comparable with the benchmarks.} 
	\label{fig:runInRome_logdet}
\end{figure*}

\begin{figure*}[ht!]
	\captionsetup[subfigure]{aboveskip=1pt}
	\centering
	\begin{subfigure}{.33\textwidth}
		\centering
 		\scalebox{0.33}{\includegraphics{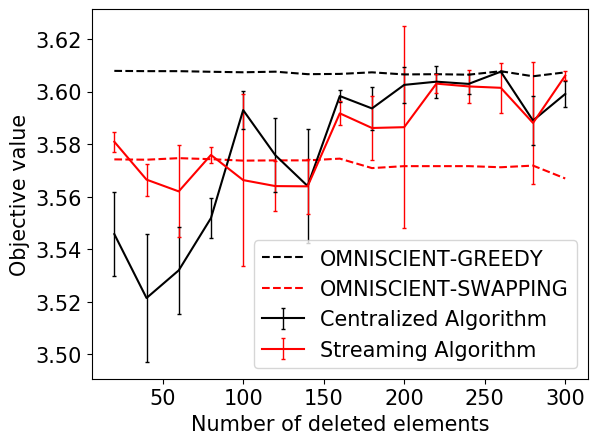}}
		\caption{\footnotesize Results for $\eps = 0.3$}
	\end{subfigure}\hspace{0.1pt}%
	\begin{subfigure}{.33\textwidth}
		\centering
		\scalebox{0.33}{\includegraphics{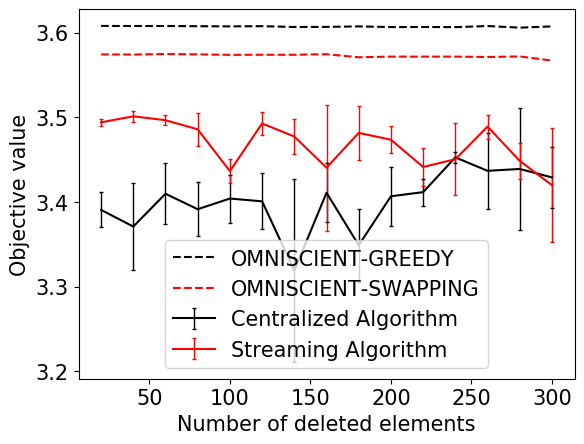}}
		\caption{\footnotesize Results for $\eps = 0.5$}
	\end{subfigure}\hspace{0.1pt}%
	\begin{subfigure}{.33\textwidth}
		\centering 		\scalebox{0.33}{\includegraphics{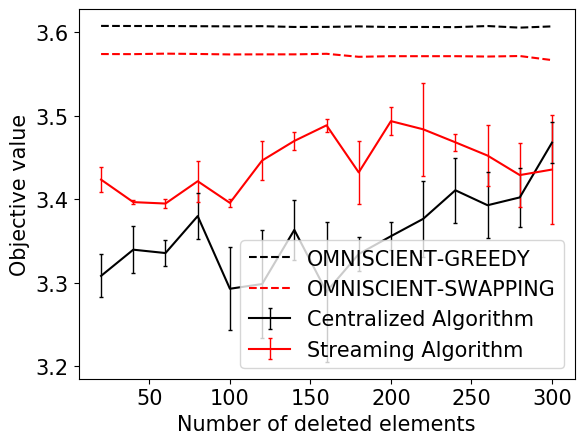}}
		\caption{\footnotesize Results for $\eps = 0.7$}
	\end{subfigure}
	\vspace{5pt}\\
	\begin{subfigure}{.33\textwidth}
		\centering
 		\scalebox{0.33}{\includegraphics{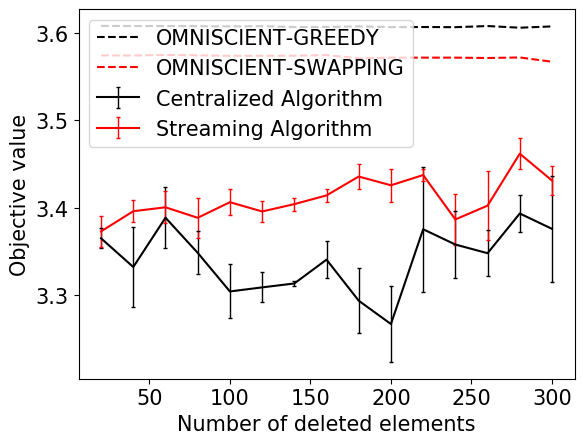}} 
        \caption{\footnotesize Results for $\eps = 0.9$}
	\end{subfigure}\hspace{0.1pt}%
	\begin{subfigure}{.33\textwidth}
		\centering
		\scalebox{0.33}{\includegraphics{runInRome kmedoid/run_kmedoid_objective_0.99.png}} 
		\caption{\footnotesize Results for $\eps = 0.99$}
	\end{subfigure}
% 	\vspace{8pt}%
	\caption{\small Results of the $k$-medoid experiment on the RunInRome dataset for different values of $\eps$. Note how as $\eps$ decreases the performances of our algorithms improve, starting from a worst case of $\approx 10 \% $ for $\eps = 0.99$ and a small number of deletions to a nearly identical performance for $\eps = 0.3.$} 
	\label{fig:runInRome_kmedoid}
\end{figure*}

\begin{figure*}[ht!]
	\captionsetup[subfigure]{aboveskip=1pt}
	\centering
	\begin{subfigure}{.33\textwidth}
		\centering
 		\scalebox{0.33}{\includegraphics{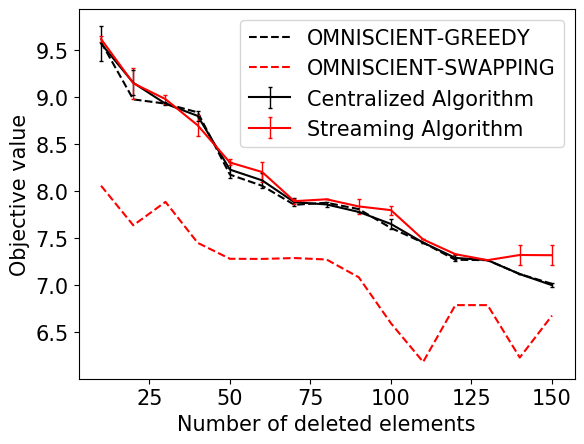}}
		\caption{\footnotesize Results for $\eps = 0.3$}
	\end{subfigure}\hspace{0.1pt}%
	\begin{subfigure}{.33\textwidth}
		\centering
		\scalebox{0.33}{\includegraphics{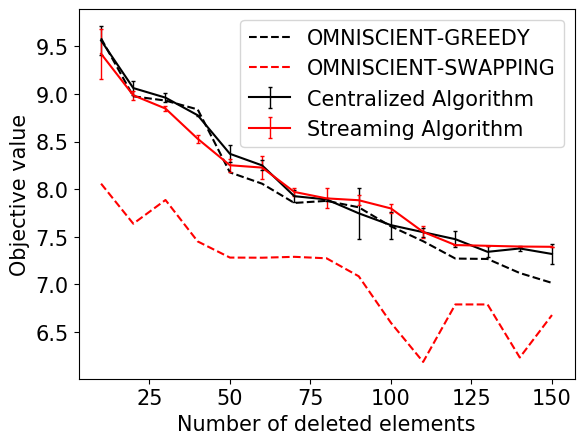}}
		\caption{\footnotesize Results for $\eps = 0.5$}
	\end{subfigure}\hspace{0.1pt}%
	\begin{subfigure}{.33\textwidth}
		\centering 		\scalebox{0.33}{\includegraphics{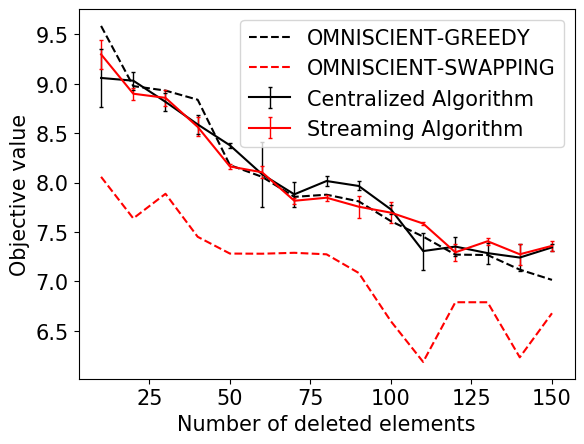}}
		\caption{\footnotesize Results for $\eps = 0.7$}
	\end{subfigure}
	\vspace{5pt}\\
	\begin{subfigure}{.33\textwidth}
		\centering
 		\scalebox{0.33}{\includegraphics{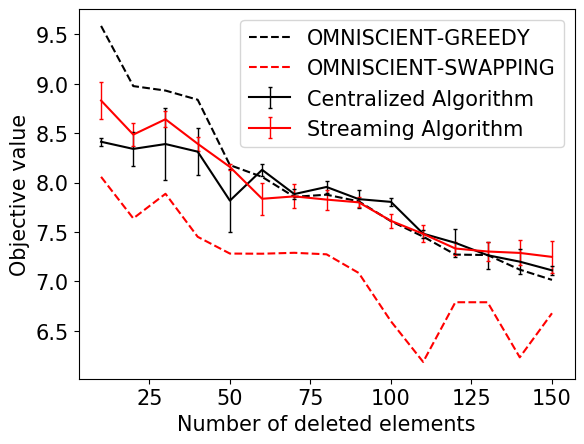}} 
        \caption{\footnotesize Results for $\eps = 0.9$}
	\end{subfigure}\hspace{0.1pt}%
	\begin{subfigure}{.33\textwidth}
		\centering
		\scalebox{0.33}{\includegraphics{uber logdet/uber_logdet_objective_0.99.png}} 
		\caption{\footnotesize Results for $\eps = 0.99$}
	\end{subfigure}
% 	\vspace{8pt}%
	\caption{\small Results of the kernel logdet experiment on the Uber dataset for different values of $\eps$. Note that there is never a remarkable difference between our algorithms and the strongest benchmark \omngreedy, but as $\eps$ decreases also the variability in the output of our algorithms decreases.} 
	\label{fig:uber_logdet}
\end{figure*}

\begin{figure*}[ht!]
	\captionsetup[subfigure]{aboveskip=1pt}
	\centering
	\begin{subfigure}{.33\textwidth}
		\centering
 		\scalebox{0.33}{\includegraphics{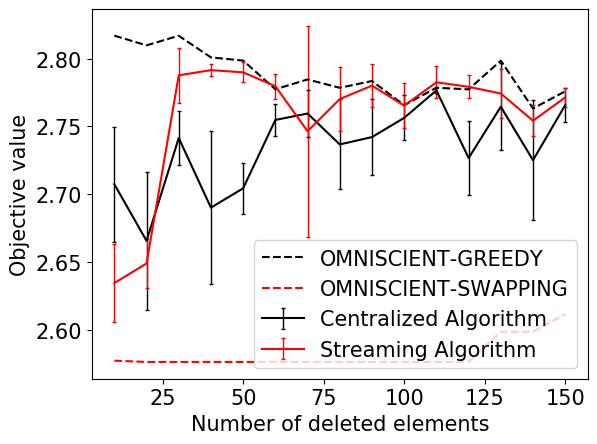}}
		\caption{\footnotesize Results for $\eps = 0.3$}
	\end{subfigure}\hspace{0.1pt}%
	\begin{subfigure}{.33\textwidth}
		\centering
		\scalebox{0.33}{\includegraphics{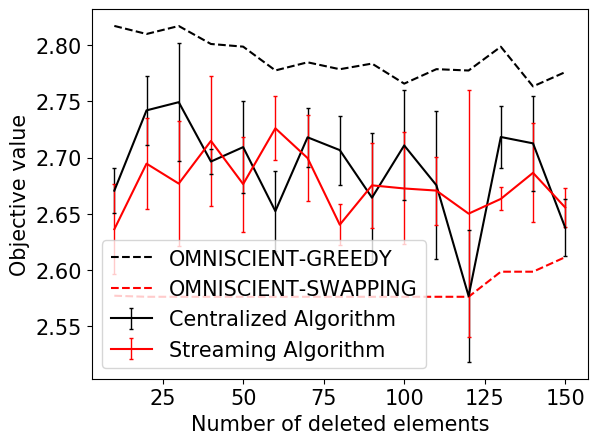}}
		\caption{\footnotesize Results for $\eps = 0.5$}
	\end{subfigure}\hspace{0.1pt}%
	\begin{subfigure}{.33\textwidth}
		\centering 		\scalebox{0.33}{\includegraphics{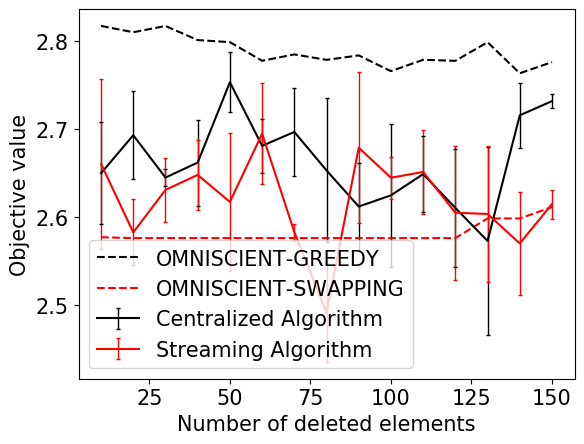}}
		\caption{\footnotesize Results for $\eps = 0.7$}
	\end{subfigure}
	\vspace{5pt}\\
	\begin{subfigure}{.33\textwidth}
		\centering
 		\scalebox{0.33}{\includegraphics{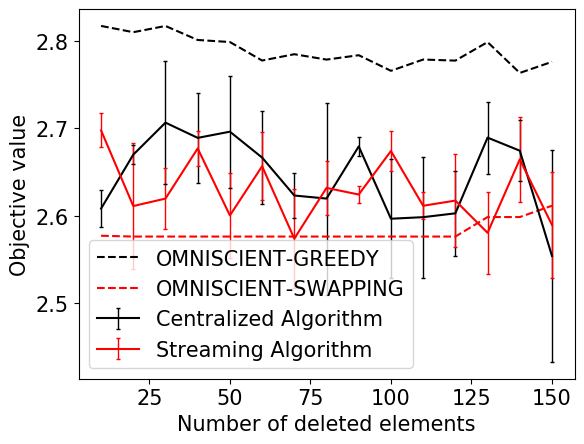}} 
        \caption{\footnotesize Results for $\eps = 0.9$}
	\end{subfigure}\hspace{0.1pt}%
	\begin{subfigure}{.33\textwidth}
		\centering
		\scalebox{0.33}{\includegraphics{uber kmedoid/uber_kmedoid_objective_0.99.png}} 
		\caption{\footnotesize Results for $\eps = 0.99$}
	\end{subfigure}
% 	\vspace{8pt}%
	\caption{\small Results of the $k$-medoid experiment on the Uber dataset for different values of $\eps$. Note how as $\eps$ decreases the performances of our algorithms improve, while being always comparable with the benchmarks.} 
	\label{fig:uber_kmedoid}
\end{figure*}

\end{document}